\documentclass[aps,prd,floats,floatfix,superscriptaddress,nofootinbib,noshowpacs]{revtex4-2}

\usepackage{hyperref}
\usepackage{amssymb}
\usepackage{amsmath}
\usepackage{physics}
\usepackage{latexsym}

\usepackage{verbatim}
\usepackage{mathrsfs}
\usepackage{amsfonts}
\usepackage{graphicx,subfigure}
\usepackage{epsfig}

\usepackage{soul}
\usepackage{ulem}
\usepackage{xcolor}
\usepackage{multirow}

\newcommand{\Natural}{\mathbb{N}}
\newcommand{\Real}{\mathbb{R}}
\newcommand{\Complex}{\mathbb{C}}
\newcommand{\identity}{1\!\! 1}
\newcommand{\re}{\mbox{Re}}

\newcommand{\proof}{\noindent {\bf Proof. }}
\newcommand{\proofof}[1]{\noindent {\bf Proof of #1. }}
\newcommand{\qed}{\hfill \fbox{} \vspace{.3cm}}

\newtheorem{lemma}{Lemma}
\newtheorem{proposition}{Proposition}
\newtheorem{corollary}{Corollary}
\newtheorem{theorem}{Theorem}
\newtheorem{remark}{Remark}

\begin{document}

\title{Existence of nonrelativistic $\ell$- and multi-$\ell$-boson stars and their radial stability}

\author{Emmanuel Ch\'avez Nambo}
\affiliation{Instituto de F\'isica y Matem\'aticas,
Universidad Michoacana de San Nicol\'as de Hidalgo,
Edificio C-3, Ciudad Universitaria, 58040 Morelia, Michoac\'an, M\'exico}
\author{Olivier Sarbach}
\affiliation{Instituto de F\'isica y Matem\'aticas,
Universidad Michoacana de San Nicol\'as de Hidalgo,
Edificio C-3, Ciudad Universitaria, 58040 Morelia, Michoac\'an, M\'exico}

\date{\today}

\begin{abstract}
Using direct methods of the calculus of variations we establish the existence of an infinite class of spherically-symmetric solutions to the multi-field Schr\"odinger-Poisson system. This is achieved by proving that the energy functional admits a global minimum when restricted to the set of vector-valued wave functions in the Sobolev space $H^1$ which are invariant with respect to a suitable representation of the rotation group and whose components have fixed $L^2$-norms. Additionally, we show that these minima correspond to solutions which are orbital stable with respect to perturbations of the wave function within this set. The generalization to include an external potential and some important properties of the minima are also discussed.
\end{abstract}

\maketitle

\section{Introduction}
\label{Sec:Intro}

Recently, there has been renewed interest in stationary solutions of the multi-field Schr\"odinger-Poisson system, mainly motivated from the possibility that dark matter is described by an ultra-light scalar or vector field~\cite{PhysRevLett.85.1158, Matos:2000ss, Marsh:2010wq, Hui:2016ltb, FerreiraReview, Jain:2021pnk}. In particular, spherically symmetric stationary solutions play a central role since they serve as approximate models for dark matter halo cores in galaxies~\cite{Schive:2014dra,Amin:2022pzv}. Numerical investigations strongly suggest that there is a rich spectrum of such solutions, which have been referred to as multi-state boson star, $\ell$-boson stars, or multi-$\ell$-boson stars in the literature (see Refs.~\cite{Bernal:2009zy,Alcubierre:2018ahf,SemiclassicalBosonStars} for the relativistic regime in which they were initially proposed and~\cite{Urena-Lopez:2010zva,Guzman:2019gqc,Roque:2023sjl,Navarro-Boullosa:2023bya} for their nonrelativistic counterparts). As shown in~\cite{Nambo:2023yut}, $\ell$-boson stars can be characterized by the property that their wave function is invariant with respect to a suitable representation of the rotation group. Based on a  combination of analytic and numerical methods, $\ell$-boson stars whose radial profile have no nodes have been argued to be stable with respect to radial linear perturbations~\cite{Roque:2023sjl}. Furthermore, these solutions have been conjectured to be stable with respect to arbitrary linear perturbations when $\ell=0$ or $\ell=1$, see Ref.~\cite{Nambo:2023yut}.

The purpose of this article is to provide a rigorous proof for the existence of (multi-) $\ell$-boson stars as  solutions of the multi-field Schr\"odinger-Poisson system. More specifically, we show that the solutions which have been conjectured to be stable with respect to radial linear perturbations can be obtained by minimizing the total energy functional on a suitable function space $H^1_\rho$, consisting of wave functions which are invariant under the appropriate representation of the rotation group, keeping fixed the $L^2$-norm of their components. Furthermore, we prove that each of these solutions are orbital stable with respect to  perturbations within the space $H^1_\rho$. In particular, this establishes the stability of $\ell$- and multi-$\ell$-boson stars with respect to nonlinear radial perturbations.

The Schr\"odinger-Poisson system also plays an important role in other areas of physics, including Choquard's one-component plasma model~\cite{Lieb1977,LionsChoquard,LIONS1984109,ChoquardGuide} and the gravitationally induced collapse of the wave function~\cite{Moroz}. For applications to the pseudo-relativistic and fully relativistic theory of boson stars, see for instance~\cite{LiebYau2005,Frohlich:2005sh} and~\cite{Bizon:2000es,Liebling:2012fv}, respectively.

This article is organized as follows. In section~\ref{Sec:MainResults} we review previous results and state the main theorems of this article regarding the existence and orbital stability of $\ell$- and multi-$\ell$ boson stars. Next, in section~\ref{Sec:Preliminaries} we discuss some preliminary estimates and results. Section~\ref{Sec:Min} is devoted to the proof for the existence of the global minima corresponding to $\ell$-boson stars, based on Strauss' lemma which ultimately allows one to extract a convergent subsequence from a minimizing sequence in $H^1_\rho$. Next, in section~\ref{Sec:Properties} we discuss the main properties of these minima, including virial relations, scaling properties, and the fact that they constitute strong solutions of the multi-field Schr\"odinger-Poisson system. Section~\ref{Sec:MutliEll} discusses the existence and main properties of multi-$\ell$-boson stars. Finally, in section~\ref{Sec:OrbitalStability} we show that the solutions are orbital stable in $H^1_\rho$. For completeness, we also include in appendix~\ref{App:GlobalExistence} a proof for the global in time well-posedness of the Cauchy problem for the multi-field Schr\"odinger-Poisson system, whereas appendix~\ref{App:Projector} discusses the properties of a relevant projection operator.

We use the following notations and definitions: $L^p(\Real^3,\Complex^m)$ with $p\geq 1$ denotes the space of (equivalent classes of) measurable functions $u: \Real^3\to\Complex^m$ such that $|u|^p$ is Lebesgue-integrable. The corresponding norm is denoted by $\| u \|_p$. For each $s=0,1,2,\ldots$, $H^s(\Real^3,\Complex^m)$ is the Sobolev space consisting of functions $u\in L^2(\Real^3,\Complex^m)$ whose  partial derivatives $D^\alpha u$ of order $|\alpha|\leq s$ also belong to $L^2(\Real^3,\Complex^m)$, with the corresponding norm
\begin{equation}
\| u \|_{H^s} := \sqrt{\sum\limits_{|\alpha|\leq s} \| D^\alpha u \|_2^2},\qquad
u\in H^s(\Real^3,\Complex^m).
\end{equation}

\section{Main results}
\label{Sec:MainResults}

The multi-field Schr\"odinger-Poisson system is given by
\begin{align}
i\frac{\partial \psi}{\partial t} &= -\Delta\psi + (U + V)\psi,
\label{Eq:Schroedinger}\\
\Delta U &= |\psi|^2,
\label{Eq:Poisson}
\end{align}
where $\psi: \Real\times\Real^3\to \Complex^m$, $(t,x)\mapsto \psi(t,x)$, is a vector-valued wave function with $m$ components, $\psi = (\psi_1,\psi_2,\ldots,\psi_m)$, $\Delta$ the Laplace operator on $\Real^3$ (acting on each component of $\psi$) and $U: \Real\times\Real^3\to \Real$, $(t,x)\mapsto U(t,x)$ refers to the gravitational potential induced by the field $\psi$. Further, $V: \Real^3\to \Real$, $x\mapsto V(x)$ is an external, time-independent potential. For the physical scenario of a self-gravitating dark matter cloud (whose particle number density is described by $|\psi|^2 = \sum_{j=1}^m |\psi_j|^2$), the external potential $V$ models the gravitational potential induced by the galactic disk of baryonic matter or by the supermassive black hole in the center of the galaxy, treating the black hole as a point mass~\cite{ElliotMocz}. In the latter case, $V(x) = V_{BH}(x) = -M/|x|$ with $M$ the black hole's mass. An example concerning baryonic matter is the Miyamoto-Nagai potential $V = V_{MN}$, given by~\cite{mMrN75}
\begin{equation}
V_{MN}(x) = -\frac{M}{\sqrt{x_1^2 + x_2^2 + \left( a + \sqrt{x_3^2 + b^2} \right)^2}},\qquad x = (x_1,x_2,x_3)\in\Real^3,
\label{Eq:MNPotential}
\end{equation}
with $M$ the total baryon mass and $a$ and $b$ two positive parameters. In the limit $a\to 0$ this potential reduces to the Plummer potential which is rotational symmetric, and when both $a$ and $b$ tend to zero one recovers $V_{BH}$.

One can check that the system~(\ref{Eq:Schroedinger},\ref{Eq:Poisson}) formally preserves the positive semi-definite Hermitian matrix
\begin{equation}
Q[\psi] := \int \overline{\psi} \otimes \psi dx
\label{Eq:QDef}
\end{equation}
and the total energy
\begin{equation}
\mathcal{E}[\psi] := \frac{1}{2}\int
 |\nabla\psi|^2 dx + \frac{1}{4}\int U(t,x)|\psi|^2 dx + \frac{1}{2} \int V(x)|\psi|^2 dx,
\end{equation}
where here and in the following, $\overline{\psi}$ denotes the complex conjugate of $\psi$, such that $|\psi|^2 = \overline{\psi}\cdot \psi$, and the integrals are performed over $\Real^3$. Assuming that $U(t,x)\to 0$ for $|x|\mapsto \infty$ it is convenient to represent the gravitational potential in the form
\begin{equation}
U(t,x) = \Delta^{-1}(|\psi|^2)(t,x) := -\frac{1}{4\pi}\int \frac{|\psi(t,y)|^2}{|x-y|} dy,
\label{Eq:GravPotential}
\end{equation}
and work with the resulting integro-differential equation obtained from Eq.~(\ref{Eq:Schroedinger}). Accordingly, the energy functional is written in the form 
\begin{equation}
\mathcal{E}[u] = T[u] - L[n] - D[n],
\label{Eq:EnergyFunctional}
\end{equation}
where $T$, $L$, and $D$ are defined as 
\begin{align}
T[u] &:= \frac{1}{2}\int
 |\nabla u(x)|^2 dx,
\label{Eq:DefT}\\
L[n] &:= \frac{1}{2}\int (-V(x)) n(x) dx,
\label{Eq:DefL}\\
D[n] &:= \frac{1}{16\pi} \int\int \frac{n(x) n(y)}{|x-y|} dx dy,
\label{Eq:DefD}
\end{align}
and $n := |u|^2$.

For the following, we make the following assumptions on $V$:
\begin{enumerate}
\item[(i)] $V:\Real^3\to \Real$ is a measurable function satisfying $V\in L^2_{\rm{loc}}(\Real^3)$ and $V(x)\to 0$ as $|x|\to \infty$,
\item[(ii)] $V\leq 0$,
\end{enumerate}
which include the aforementioned examples $V_{BH}$ and $V_{MN}$ when $M\geq 0$. As explicitly shown later (see Lemma~\ref{Lem:BasicEstimates} and Lemma~\ref{Lem:H0Properties}), these conditions imply that $T$, $L$, $D$, and $\mathcal{E}$ are well-defined for each $u\in H^1(\Real^3,\Complex^m)$ and that the Schr\"odinger operator $H_0 := -\Delta + V: D(H_0)\subset L^2(\Real^3,\Complex^m)\to L^2(\Real^3,\Complex^m)$ is self-adjoint with domain $D(H_0) = H^2(\Real^3,\Complex^m)$ and bounded from below.

\subsection{The Cauchy problem}

The Cauchy problem associated with Eq.~(\ref{Eq:Schroedinger}) can be written as
\begin{eqnarray}
i\frac{d\psi}{dt}(t) &=& H_0\psi(t) + F(\psi(t)),\qquad t\in\Real,
\label{Eq:Cauchy1}\\
\psi(0) &=& \psi_0,
\label{Eq:Cauchy2}
\end{eqnarray}
with the self-adjoint operator $H_0 = -\Delta + V$ defined above and the nonlinear map $F(\psi) := \Delta^{-1}(|\psi|^2) \psi$, with $\Delta^{-1}$ defined as in Eq.~(\ref{Eq:GravPotential}). Since $F: H^2(\Real^3,\Complex^m)\to H^2(\Real^3,\Complex^m)$ is a well-defined, locally Lipschitz continuous map (see~\cite{Lenzmann} and appendix~\ref{App:GlobalExistence}), one can prove that the Cauchy problem~(\ref{Eq:Cauchy1},\ref{Eq:Cauchy2}) is local in time well-posed. In fact, one can prove that it is also global in time well-posed:

\begin{theorem}
\label{Thm:WP}
Suppose $V$ satisfies condition (i). Then, given $\psi_0\in D(H_0)$ there exists a unique, continuously differentiable curve $\psi: (-\infty,\infty)\to L^2(\Real^3,\Complex^m)$ such that $\psi(t)\in D(H_0)$ for all $t\in\Real$, $\psi: (-\infty,\infty)\to H^2(\Real^3,\Complex^m)$ is continuous, and $\psi$ satisfies Eqs.~(\ref{Eq:Cauchy1},\ref{Eq:Cauchy2}). Moreover, $\mathcal{E}[\psi(t)] = \mathcal{E}[\psi_0]$ and $Q[\psi(t)] = Q[\psi_0]$ for all $t\in \Real$.
\end{theorem}

It is also possible to extend the solution concept to the space $H^1(\Real^3,\Complex^m)$. For the case $m=1$, see for instance~\cite{cazenave2003semilinear} and~\cite{Ginibre1980OnAC} where  Theorem~\ref{Thm:WP} was first proven for $V=0$. For completeness, we include a proof of Theorem~\ref{Thm:WP} for arbitrary $m\in \Natural$ in appendix~\ref{App:GlobalExistence} which is based on semi-group methods and some key estimates for the map $F$ following the work in Ref.~\cite{Lenzmann}.

\subsection{Stationary ground state solution and orbital stability}

One of the main results of this article concerns the existence of {\it stationary solutions}, for which the wave function has the form $\psi(t,x) = e^{-i \omega t} u(x)$ with $\omega\in\Real$, and their stability. These solutions satisfy the stationary problem
\begin{equation}
H_0 u + F(u) = \omega u,
\label{Eq:SPTimeIndependent}
\end{equation}
which can be interpreted as a nonlinear eigenvalue problem for $\omega$, where $u$ must lie in a suitable function space. More specifically, we are interested in those stationary solutions that can be obtained by minimizing the energy functional $\mathcal{E}[u]$ over a suitable function space, keeping fixed the matrix $Q$ defined in Eq.~(\ref{Eq:QDef}). 

First, let us summarize what is known when only one field is present ($m=1$), in which case $Q[\psi] = N$ is the total particle number. Define for each $N > 0$ the quantity
\begin{equation}
E(N) := \inf\left\{ \mathcal{E}[u] : u\in H^1(\Real^3,\Complex), \| u \|_2^2 = N \right\}.
\end{equation}
From the pioneering works by H. Lieb and P.S. Lions~\cite{Lieb1977,LIONS1984109} one can deduce:

\begin{theorem}
\label{Thm:GroundState}
Let $N > 0$ and suppose $V$ satisfies $V\in L^p(\Real^3) + L^q(\Real^3)$ with $3/2\leq p,q < \infty$ and assumption (ii). Then, $-\infty < E(N) < 0$. Moreover, $E(N)$ possesses a minimizer, that is, there exists $u_*\in H^1(\Real^3,\Complex)$ with $\| u_* \|_2^2 = N$, such that
\begin{equation}
\mathcal{E}[u_*] = E(N).
\end{equation}
Additionally,
\begin{itemize}
\item[(a)] $u_*$ is a strong solution of the problem~(\ref{Eq:SPTimeIndependent}). That is, $u_*\in H^2(\Real^3,\Complex)$ and there exists $\omega\in\Real$ such that Eq.~(\ref{Eq:SPTimeIndependent}) is satisfied.
\item[(b)] $E(N)$ and $\omega$ are related according to
\begin{equation}
\frac{1}{2} N\omega = E(N) - D[n_*],\qquad n_* := |u_*|^2.
\label{Eq:Virial2}
\end{equation}
\item[(c)] If $V$ is rotational symmetric, that is, if $V(Rx) = V(x)$ for all $R\in SO(3)$ and almost all $x\in\Real^3$, then $u_*$ is rotational symmetric and has no zeros. If $V\equiv 0$, then $u_*$ is unique, up to translations and a multiplication with a constant phase.
\item[(d)] Let $(u_n)$ be a minimizing sequence, that is, $u_n\in H^1(\Real^3,\Complex)$ satisfies $\| u_n \|_2^2 = N$ and $\mathcal{E}[u_n]\to E(N)$. Then, $(u_n)$ possesses a subsequence which converges strongly in $H^1(\Real^3,\Complex)$ to a minimizer.
\end{itemize}
\end{theorem}

\begin{remark}
It follows from the property $E(N) < 0$ and the relation~(\ref{Eq:Virial2}) that $\omega < 0$.
\end{remark}

For rotational-symmetric $V$ the proof follows directly from the results in~\cite{Lieb1977} based on symmetric rearrangements. For the more general case for which $V$ is not necessarily rational symmetric, see~\cite{LIONS1984109} which uses the concentration-compactness principle instead. 

The significance of property (d) relies in the fact that one can prove the {\it orbital stability} of minimizers. More precisely, let
\begin{equation}
S_N := \left\{ u\in H^1(\Real^3,\Complex) : \| u \|_2^2 = N, \mathcal{E}[u] = E(N)  \right\}
\end{equation}
be the subset of minimizers (which, according to Theorem~\ref{Thm:GroundState}, is non-empty). The set $S_N$ is called orbital stable if given $\varepsilon > 0$ there exists $\delta > 0$ such that for $\psi_0\in H^2(\Real^3,\Complex)$ satisfying
\begin{equation}
\mbox{dist}(\psi_0,S_N) := \inf\limits_{u\in S_N} \| \psi_0 - u \|_{H^1} < \delta
\end{equation}
the unique global solution $\psi: (-\infty,\infty)\to H^2(\Real^3,\Complex)$ of the Cauchy problem~(\ref{Eq:Cauchy1},\ref{Eq:Cauchy2}) satisfies
\begin{equation}
\mbox{dist}(\psi(t),S_N) < \varepsilon
\end{equation}
for all $t\in \Real$.

As follows from Ref.~\cite{Cazenave1982}, one has:

\begin{theorem}
Under the same hypotheses as in the previous theorem, the set $S_N$ is orbital stable in $H^1(\Real^3,\Complex)$.
\end{theorem}

\subsection{Nonrelativistic $\ell$-boson stars}

In the present work we construct minimizers of the total energy functional $\mathcal{E}$ on suitable subspaces of $H^1(\Real^3,\Complex^m)$ and prove that these minimizers are orbital stable {\it within these subspaces}. For this, we exploit the fact that the rotation group $SO(3)$ can act in a non-trivial way on the components of the wave function. More precisely, consider a specific unitary representation $D: SO(3)\to GL(\Complex^m)$, $R\mapsto D(R)$, of the rotation group on the vector space $\Complex^m$. Then, we consider the subspace $H^1_\rho$ of $H^1(\Real^3,\Complex^m)$ consisting of functions $u$ which are invariant with respect to the unitary representation $\rho: SO(3)\to GL( L^2(\Real^3,\Complex^m))$ defined by
\begin{equation}
[\rho(R) u](x) := D(R) u(R^{-1} x),\qquad
R\in SO(3),\quad x\in \Real^3,\quad
u\in H^1(\Real^3,\Complex^m),
\label{Eq:Representation}
\end{equation}
that is,
\begin{equation}
H^1_\rho := \{ u\in H^1(\Real^3,\Complex^m) : \rho(R) u = u \hbox{ for all $R\in SO(3)$} \}.
\label{Eq:H1rho}
\end{equation}
Note that $H^1_\rho$ is a closed subspace of $H^1(\Real^3,\Complex^m)$ and hence it is a Hilbert space (with the scalar product and norm inherited from $H^1$). It is easy to verify that for $u\in H^1_\rho$ the quantity $Q$ defined in Eq.~(\ref{Eq:QDef}) satisfies
\begin{equation}
[D(R),Q] = 0\quad \hbox{for all $R\in SO(3)$.}
\label{Eq:Commutator}
\end{equation}

We start with the case in which the unitary representation $D = D_\ell$ is {\it irreducible}, where $\ell\in \Natural_0$ labels the corresponding angular momentum number and $m = 2\ell+1$. According to Schur's lemma (see, for example, Ref.~\cite{Hall-Groups}), it follows from Eq.~(\ref{Eq:Commutator}) that $Q$ is proportional to the identity operator on $\Complex^{2\ell+1}$, that is, $Q = K\identity_{2\ell+1}$ with $N = \Tr(Q) = K(2\ell+1)$. Therefore, we define for each $N > 0$ the quantity
\begin{equation}
E_\ell(N) := \inf\left\{  \mathcal{E}[u] : u\in H^1_\rho, \| u \|_2^2 = N \right\}.
\end{equation}
The first main result of this article is:

\begin{theorem}
\label{Thm:Main}
Let $D_\ell: SO(3)\to GL(\Complex^{2\ell+1})$ be an irreducible unitary representation of the rotation group on the vector space $\Complex^{2\ell+1}$ and consider the associated invariant subspace $H^1_\rho$ defined in Eq.~(\ref{Eq:H1rho}). Suppose the external potential $V$ satisfies the assumptions (i) and (ii) and is rotational symmetric, and let $N > 0$. Then, $-\infty < E_\ell(N) < 0$ and $E_\ell(N)$ possesses a minimizer $u_*$.

Moreover, $u_*\in H^2(\Real^3,\Complex^{2\ell+1})$ is a strong solution of Eq.~(\ref{Eq:SPTimeIndependent}), satisfies property (b) of Theorem~\ref{Thm:GroundState} (with $E(N)$ replaced with $E_\ell(N)$), and the set
\begin{equation}
S_{\ell,N} := \left\{ u\in H^1_\rho : \| u \|_2^2 = N, \mathcal{E}[u] = E_\ell(N)  \right\}
\label{Eq:SrhoN}
\end{equation}
is orbital stable in $H^1_\rho$.
\end{theorem}

\begin{remark}
The minimizer yields the ``ground state" nonrelativistic $\ell$-boson star discussed in~\cite{Roque:2023sjl}. It has the form
\begin{equation}\label{Eq:ExplicitFormEll}
u_*(x) = f_\ell(r) \mathcal{Y}^\ell(\hat{x}),\qquad
r = |x|,\quad \hat{x} = \frac{x}{r},
\end{equation}
where the radial function $f_\ell: (0,\infty)\to \Real$ is nodeless and $\mathcal{Y}^\ell = (Y^{\ell,-\ell},Y^{\ell,1-\ell},\ldots,Y^{\ell,\ell})$ with $Y^{\ell,m}$ denoting the standard spherical harmonics.

Moreover, the numerical investigation in~\cite{Roque:2023sjl} indicates that this solution is mode stable for linear perturbations which are invariant with respect to $\rho$, meaning that such perturbations possess only purely oscillatory modes. Theorem~\ref{Thm:Main} strengthens these results by proving that ground state nonrelativistic $\ell$-boson stars are, in fact, orbital stable within $H^1_\rho$.
\end{remark}
\begin{remark}
For the case $V \equiv 0$, the existence of ground and excited states was rigorously established in~\cite{2002math.ph...8045H} and \cite{nambo21} for $\ell=0$ and $\ell\geq 1$, respectively, by analyzing directly the radial problem for $f_\ell$. In particular, it was shown that for each $\ell\in \Natural_0$ there exists a family of functions $f_{n\ell}$ ($n \in \mathbb{N}_0$) such that the corresponding $\ell$-boson star configurations $u_{n\ell}(x) = f_{n\ell}(r) \mathcal{Y}_\ell(\hat{x})$ solve the stationary problem~\eqref{Eq:SPTimeIndependent}, with the index $n$ characterizing the number of nodes of the radial profile $f_{n\ell}$.
\end{remark}
\begin{remark}
Theorem~\ref{Thm:Main} {\it does not} imply that the minimizers in $E_\ell(N)$ are orbital stable in the larger space $H^1(\Real^3,\Complex^{2\ell+1})$. In fact, the numerical results in~\cite{Nambo:2023yut} suggest that nonrelativistic $\ell$-boson stars with $\ell > 1$ are linearly unstable.
\end{remark}

For the particular case $V = V_{BH}$ one can further show:

\begin{theorem}
\label{Thm:BHPotential}
Suppose $V(x) = V_{BH}(x) = -M/|x|$, and denote for each $M\geq 0$ and $N > 0$ by $E_\ell(M,N)$ and $\omega_\ell(M,N)$ the energy and frequency belonging to the minimizer $u_*$ in Theorem~\ref{Thm:Main}. Then,
\begin{enumerate}
\item[(a)] the following virial relation holds:
\begin{equation}
2T[u_*] - L[n_*] - D[n_*] = 0,\qquad
n_* = |u_*|^2,
\label{Eq:Virial}
\end{equation}
which implies $E_\ell(M,N) = -T[u_*]$ and $\omega_\ell(M,N) = 2N^{-1}(3E_\ell(M,N) + L[n_*])$.
\item[(b)] the functions $E_\ell,\omega_\ell: [0,\infty)\times (0,\infty)\to \Real$ satisfy the following scaling relations for all $M\geq 0$ and $N > 0$:
\begin{equation}
E_\ell(M,N) = N^3 E_\ell(N^{-1} M,1),\qquad
\omega_\ell(M,N) = N^2\omega_\ell(N^{-1} M,1).
\label{Eq:Scaling}
\end{equation}
In particular when $M=0$, one has $E_\ell(0,N) = N^3 E_\ell(0,1)$, $\omega_\ell(0,N) = N^2\omega_\ell(0,1)$, and $\omega_\ell(0,1) = 6E_\ell(0,1)$.
\item[(c)] For fixed $M > 0$ one has
\begin{equation}
\lim\limits_{N\to 0} \omega_\ell(M,N) = E_\ell^{(0)} = -\frac{M^2}{4(\ell+1)^2},
\qquad
\lim\limits_{N\to\infty} \frac{\omega_\ell(M,N) - \omega_\ell(0,N)}{N^2} = 0,
\label{Eq:MLimits}
\end{equation}
where $E_\ell^{(0)}$ is the ground state energy of the Hamilton operator $H_0 = -\Delta + V_{BH}$ restricted to $H_\rho^1$.
\end{enumerate}
\end{theorem}

\subsection{Nonrelativistic multi-$\ell$-boson stars}

Next, consider the case for which the unitary representation $D: SO(3)\to GL(\Complex^m)$ is {\it reducible}. More specifically, we assume that $D$ has the following form:
\begin{equation}
D = D_{\ell_1} \oplus D_{\ell_2} \oplus \cdots \oplus D_{\ell_r},
\label{Eq:DDecomposition}
\end{equation}
with {\it inequivalent} irreducible components $D_{\ell_j}$. Thus, without loss of generality, $0\leq \ell_1 < \ell_2 < \ldots < \ell_r$ and $r\geq 2$.  Schur's lemma and Eq.~(\ref{Eq:Commutator}) imply that $Q$ must have the following diagonal form:
\begin{equation}
Q = \left( \begin{array}{cccc}
K_1\identity_{2\ell_1+1} & & & \\
& K_2\identity_{2\ell_2+1} & & \\
& & \ddots &  \\
& & & K_r\identity_{2\ell_r+1}
\end{array} \right),
\label{Eq:QFormReducible}
\end{equation}
with $K_j\geq 0$. For what follows, it is convenient to consider the corresponding decomposition $H^1_\rho = \bigoplus_{j=1}^{r} H^1_{\ell_j}$, such that a function in $H^1_\rho$ is written as $u = (u^{(1)}, u^{(2)}, \dots, u^{(r)})$, where each $u^{(j)}$ belongs to the space $H^1_{\ell_j}$ defined by
\begin{equation}
H^1_{\ell_j} := \{ u\in H^1(\Real^3,\Complex^{2\ell_j+1}) : \rho_{\ell_j}(R) u^{(j)} = u^{(j)} \hbox{ for all $R\in SO(3)$} \},
\end{equation}
where 
\begin{equation}
[\rho_{\ell_j}(R) u^{(j)}](x) = D_{\ell_j}(R)u^{(j)}(R^{-1}x),\qquad
R\in SO(3),\quad x\in\Real^3.
\end{equation}
Note that the matrix $Q[u]$ associated with a function $u\in H^1_\rho$ has the form~(\ref{Eq:QFormReducible}), where $\| u^{(j)} \|_2^2 = K_j(2\ell_j+1)$ for all $j=1,2,\ldots, r$. Accordingly, we define for each $N_1,N_2,\ldots,N_r > 0$ the quantity
\begin{equation}
E_{\rho}(N_1,N_2,\ldots,N_r) := \inf\left\{  \mathcal{E}[u] : u\in H^1_\rho, \| u^{(j)} \|_2^2 = N_j \text{ for each } j = 1, 2, \ldots, r \right\}.
\end{equation}
The second main result of this article is:

\begin{theorem}
\label{Thm:Main2}
Let $D: SO(3)\to GL(\Complex^m)$ be a unitary representation of the rotation group on the vector space $\Complex^m$ of the form~(\ref{Eq:DDecomposition}) with $0\leq \ell_1 < \ell_2 < \ldots < \ell_r$, and consider the associated invariant subspace $H^1_\rho$ defined in Eq.~(\ref{Eq:H1rho}). Suppose the external potential $V$ satisfies the assumptions (i) and (ii) and is rotational symmetric, and let $N_1,N_2,\ldots,N_r > 0$. Then, $-\infty < E_{\rho}(N_1,N_2,\ldots,N_r) < 0$ and $E_{\rho}(N_1,N_2,\ldots,N_r)$ possesses a minimizer $u_*$.

Moreover, the minimizer $u_*$ satisfies the following conditions:
\begin{enumerate}
\item[(a)] $u_*\in H^2(\Real^3,\Complex^m)$ is a strong solution of the problem~(\ref{Eq:SPTimeIndependent}), where $\omega$ is replaced with a diagonal matrix of the form
\begin{equation}
\left( \begin{array}{cccc}
\omega_1\identity_{2\ell_1+1} & & & \\
& \omega_2\identity_{2\ell_2+1} & & \\
& & \ddots &  \\
& & & \omega_r\identity_{2\ell_r+1}
\end{array} \right),
\label{Eq:OmegaMatrix}
\end{equation}
with $\omega_j\in\Real$.
\item[(b)] The following relation holds:
\begin{equation}
\frac{1}{2}\sum\limits_{j=1}^r N_j\omega_j = E_\rho(N_1,N_2,\ldots,N_r) - D[n_*],
\qquad n_* := |u_*|^2.
\label{Eq:Virial2Bis}
\end{equation}
\item[(c)] The set
\begin{equation}
S_{\rho,N_1,N_2,\ldots,N_r} := \left\{ u\in H^1_\rho : \| u^{(j)} \|_2^2 = N_j, \mathcal{E}[u] = E_\rho(N_1,N_2,\ldots,N_r)  \right\}
\label{Eq:SrhoNs}
\end{equation}
is orbital stable in $H^1_\rho$.
\end{enumerate}
\end{theorem}

\begin{remark}
Theorem~\ref{Thm:Main2} proves the existence of nonrelativistic multi-$\ell$-boson stars which are orbital stable with respect to perturbations within $H^1_\rho$.  Relativistic analogues of these configurations have been constructed numerically in~\cite{SemiclassicalBosonStars}; however, their stability properties had not been analyzed so far.
\end{remark}

There are several interesting problems that are still open and that we hope to address in future work. One straightforward generalization is to include a self-interaction term of the form $\lambda |\psi|^2\psi$ to the right-hand side of Eq.~\eqref{Eq:Schroedinger}, leading to the (multi-field) Gross-Pitaevskii-Poisson system~\cite{Nambo:2024gvs,Nambo:2024hao}. It should be possible to extend our results regarding the existence of $\ell$- and multi-$\ell$ boson stars to this case, at least for repulsive self-interaction $\lambda > 0$.

Another interesting question is whether one can construct orbital stable solutions with repeated irreducible representations in Eq.~(\ref{Eq:DDecomposition}). Assuming a diagonal representation of $Q$, this requires that the corresponding components of the wave functions must be orthogonal to each other, which yields an additional constraint. For $\ell=0$, such solutions and their linear stability have been recently analyzed in~\cite{Nambo:2024hao,Nambo:2025lnu,Andrade:2026hgn} by numerical means.

A further important problem is to investigate the existence and orbital stability of other stationary and multi-frequency solutions that can be obtained by relaxing the $SO(3)$-symmetry imposed in this work. In particular, an open question is to characterize the global minima of the energy functional for fixed, given values of $Q$ without any symmetry assumptions.

Finally, numerical studies (see \cite{Guzman:2019gqc} and also~\cite{Jaramillo:2020rsv,Sanchis-Gual:2021edp} for the relativistic case) indicate that nonrelativistic $\ell$-boson stars with $\ell=1$ might be stable under nonlinear perturbations (whereas nonrelativistic $\ell$-boson stars with $\ell\geq 2$ seem to be linearly unstable, as previously mentioned). Therefore, it would be desirable to extend our orbital stability result in Theorem~\ref{Thm:Main} to generic perturbations in $H^1(\Real^3,\Complex^3)$ when $\ell=1$. Likewise, it would be interesting to analyze the stability of the nonrelativistic multi-$\ell$ boson stars of Theorem~\ref{Thm:Main2} with respect to generic perturbations in $H^1(\Real^3,\Complex^m)$.

\section{Preliminary results}
\label{Sec:Preliminaries}

We start with a few preliminary results regarding the energy functional and its constituents defined in Eqs.~(\ref{Eq:EnergyFunctional},\ref{Eq:DefT},\ref{Eq:DefL},\ref{Eq:DefD}). First, we establish the following property of $V$ which will be useful for what follows:

\begin{lemma}
\label{Lem:Potential}
Assume $V$ satisfies condition (i). Then, for each $\varepsilon > 0$ there exist real-valued functions $V_\varepsilon\in L^2(\Real^3)$ and $W_\varepsilon\in L^\infty(\Real^3)$ such that
\begin{equation}
V = V_\varepsilon + W_\varepsilon\quad
\hbox{and}\quad
\| W_\varepsilon \|_\infty\leq \varepsilon.
\end{equation}
\end{lemma}

\proof Let $\varepsilon > 0$. According to assumption (i) there exists $R > 0$ such that $|V(x)|\leq \varepsilon$ if $x\in \Real^3\setminus B_R$ lies outside the ball $B_R$ of radius $R$ centered at the origin. Denoting by $\chi_A$ the characteristic function associated with a measurable subset $A\subset \Real^3$, the desired functions can be defined as
\begin{equation}
V_\varepsilon := V \chi_{B_R},\qquad
W_\varepsilon := V \chi_{\Real^3\setminus B_R}.
\end{equation}
\qed

For the following, we abbreviate $X := H^1(\Real^3,\Complex^m)$ and $Y := L^1(\Real^3,\Complex^m)\cap L^3(\Real^3,\Complex^m)$ and equip $Y$ with the norm $\|\cdot \|_Y$ defined by
\begin{equation}
\| n \|_Y := \| n \|_1 + \| n \|_3,
\qquad n\in Y.
\end{equation}
Let $q\in [1,3]$. By interpolation,  $Y\subset L^q(\Real^3,\Complex^m)$ and $\| n \|_q\leq \| n \|_1^\alpha \| n \|_3^{1-\alpha} \leq \max\{ \alpha,1-\alpha \} \| n \|_Y$, where $\alpha\in [0,1]$ satisfies $\alpha + (1-\alpha)/3 = 1/q$. Furthermore, by Sobolev's inequality, $n =|u|^2\in Y$ if $u\in X$ and there exists a constant $C$ such that $\| n \|_Y\leq C \| u\|_{H^1}^2$ for all $u\in X$. Here and in the following, $C$ denotes a generic constant.

\begin{lemma}
\label{Lem:BasicEstimates}
Assume condition (i) for $V$ is satisfied. Then, the functionals $\mathcal{E},T: X\to \Real$ and $L,D: Y\to \Real$, are well-defined and continuous. Furthermore, they satisfy the following estimates:
\begin{enumerate}
\item[(a)] For each $\varepsilon > 0$ one has
$\left| L[n_2] - L[n_1] \right| 
\leq C_\varepsilon\| n_2 - n_1 \|_2 + \varepsilon \| n_2 - n_1 \|_1
\leq C\| n_2 - n_1 \|_Y$,
\item[(b)] $\left| D[n_2] - D[n_1] \right| \leq C\left( \| n_1 \|_{6/5} + \| n_2 \|_{6/5} \right)\| n_2 - n_1 \|_{6/5}
\leq C\left( \| n_1 \|_Y + \| n_2 \|_Y \right)\| n_2 - n_1 \|_Y$,
\item[(c)] $L[n] \leq \frac{1}{4} T[u] + C N$,
\item[(d)] $D[n] \leq \frac{1}{4} T[u] + C N^3$,
\end{enumerate}
for all $n_1,n_2\in Y$ and $u\in X$, where in (c) and (d) it is understood that $n = |u|^2$ and $N = \| u \|_2^2$.
\end{lemma}

\proof
To prove (a) we use Lemma~\ref{Lem:Potential} and H\"older's inequality to conclude
\begin{equation}
| 2L[n] | \leq \int |V(x)| |n(x)| dx
\leq \| V_\varepsilon \|_2 \| n \|_2 + \| W_\varepsilon \|_\infty \| n \|_1
\leq C_\varepsilon\| n \|_2 + \varepsilon \| n \|_1
\leq C \| n \|_Y,
\label{Eq:LEstimate}
\end{equation}
for all $n\in Y$, where $C_\varepsilon := \| V_\varepsilon \|_2$. The statement then follows from the linearity of $L$.

To prove (b) we use the Hardy-Littlewood-Sobolev (see, for example~\cite{LiebLoss}) inequality to obtain
\begin{align}
16\pi\left| D[n_2] - D[n_1] \right| 
 &\leq \int\int \frac{|n_2(x) - n_1(x)| |n_2(y)| + |n_1(x)| |n_2(y) - n_1(y) |}{|x-y|} dx dy
\nonumber\\
 &\leq C\left( \| n_1 \|_{6/5} + \| n_2 \|_{6/5} \right) \| n_2 - n_1 \|_{6/5}
\label{Eq:DEstimate}\\
 &\leq C\left( \| n_1 \|_Y + \| n_2 \|_Y \right) \| n_2 - n_1 \|_Y,
 \nonumber
\end{align}
for all $n_1,n_2\in Y$.

Regarding (c), we use Eq.~(\ref{Eq:LEstimate}) and interpolation to conclude
\begin{equation}
L[n] \leq C\left( \| n \|_1 + \| n \|_2 \right) \leq C\left( 
\| n \|_1 + \| n \|_1^{1/4} \| n \|_3^{3/4} \right)
 = C\left( N + N^{1/4} \| u \|_6^{3/2} \right),
\end{equation}
where we have used $\| n \|_1 = \| u \|_2^2 = N$ and $\| n \|_3 = \| u \|_6^2$ in the last step. Next, we use Sobolev's inequality, which implies $\| u \|_6^2\leq C \| \nabla u \|_2^2 =  2C T[u]$, and obtain
\begin{equation}
L[n] \leq C\left( N + N^{1/4} T[u]^{3/4} \right).
\end{equation}
Finally, we use the generalized Young inequality $a b\leq a^p/(p\delta^p) + (\delta b)^q/q$ for $a,b,\delta > 0$ and $1/p + 1/q = 1$ with $(p,q) = (4,4/3)$ and $\delta > 0$ small enough to obtain the statement.

The proof of (d) is similar. Using Eq.~(\ref{Eq:DEstimate}) and interpolation one obtains
\begin{equation}
D[n] \leq C\| n \|_{6/5}^2
\leq C \| n \|_1^{3/2} \| n \|_3^{1/2} 
 = C N^{3/2} \| u \|_6 \leq C N^{3/2} T[u]^{1/2} \leq \frac{1}{4} T[u] + C N^3.
\label{Eq:DnInequality}
\end{equation}
\qed

\begin{remark}
Estimates (a) and (b), together with Sobolev's inequality, imply that $\mathcal{E}$ is locally Lipschitz continuous on $X$. More precisely, there exists a constant $C$ such that
\begin{equation}
\left| \mathcal{E}[u_2] - \mathcal{E}[u_1] \right| \leq C\left( 1 + \| u_1 \|_{H^1}^3 + \| u_2 \|_{H^1}^3 \right)\| u_2 - u_1 \|_{H^1}
\label{Eq:LocalLipschitz}
\end{equation}
for all $u_1,u_2\in X$.
\end{remark}

\begin{remark}
Estimates (c) and (d) imply that for all $u\in X$,
\begin{equation}
\mathcal{E}[u] \geq \frac{1}{2} T[u] - C(N + N^3)
 \geq \frac{1}{4} \| u \|_{H^1}^2 - \left( C + \frac{1}{4} \right)(N + N^3),
\label{Eq:CoersiveBound}
\end{equation}
where $N = \| u \|_2^2$. In particular, the $H^1$-norm of $u$ is controlled by the energy and the total particle number $N = \Tr(Q)$.
\end{remark}

For the following, we define
\begin{equation}
E(N) := \inf\left\{ \mathcal{E}[u] : u\in H^1(\Real^3,\Complex^m), \| u \|_2^2 = N \right\}.
\end{equation}
As a consequence of the previous Lemma one has:

\begin{lemma}
\label{Lem:Coersive}
Assume conditions (i) and (ii) for $V$ hold. Then, the following properties are satisfied for each $N > 0$:
\begin{enumerate}
\item[(a)] $-\infty < E(N) < 0$,
\item[(b)] The functional $\mathcal{E}: X\to \Real$ is coersive on $\mathcal{A}_N := \{ u\in X : \| u \|_2^2 = N \}$, that is
\begin{equation}
\mathcal{E}[u_k]\to \infty
\end{equation}
for any sequence $(u_k)$ in $\mathcal{A}_N$ such that $\| u_k \|_{H^1}\to \infty$.
\end{enumerate}
\end{lemma}

\proof The estimate~(\ref{Eq:CoersiveBound}) shows that $\mathcal{E}$ is bounded from below and coersive on $\mathcal{A}_N$. To prove that $E(N)$ is negative, we take any $u\in \mathcal{A}_N$ (note that $u\neq 0$, otherwise $N=0$) and rescale it according to 
\begin{equation}
u_\lambda(x):=\lambda^{3/2} u(\lambda x),\qquad \lambda > 0,\qquad
x\in\Real^3.
\label{Eq:Rescaling}
\end{equation}
Since $\| u_\lambda \|_2^2 = \| u \|_2^2 = N$ it follows that $u_\lambda\in\mathcal{A}_N$ for all $\lambda > 0$. Furthermore, it is straightforward to verify that $T[u_\lambda] = \lambda^2 T[u]$ and $D[n_\lambda] = \lambda D[n]$, where $n_\lambda := |u_\lambda|^2$. Therefore, taking into account that $L\geq 0$, one obtains, for small enough $\lambda > 0$,
\begin{equation}
\lambda^{-1}\mathcal{E}[u_\lambda] = \lambda T[u] - \lambda^{-1} L[n_\lambda] - D[n] \leq \lambda T[u] - D[n] < 0.
\end{equation}
Consequently, $E(N) < 0$.
\qed

\begin{remark}
Lemma~\ref{Lem:Coersive} also holds if one replaces $E(N)$ with $E_{\ell}(N)$ and $\mathcal{A}_N$ with $\mathcal{A}_{\rho,N} := \{ u \in H^1_\rho: \| u \|_2^2 = N \}$, since the rescaling~(\ref{Eq:Rescaling}) leaves $H^1_\rho$ invariant.
\end{remark}

Let $(u_k)$ be a minimizing sequence, that is, a sequence $(u_k)$ in $\mathcal{A}_N$ such that
\begin{equation}
\mathcal{E}[u_k]\to E(N).
\label{Eq:MinimizingSequence}
\end{equation}
Then, the previous Lemma implies that it is bounded in $X$. Therefore, there exists a subsequence $(u_{k_j})$ which converges weakly to some $u_*\in X$ in $X$. It remains to show that $u_*$ is a desired minimizer. Here, the challenge is to show that one can extract a further subsequence from $(u_{k_j})$ which converges strongly in $X$, which then implies that $\mathcal{E}[u_*] = \lim_{j\to\infty}\mathcal{E}[u_{k_j}] = E(N)$. For this, a compactness argument is needed, which can be obtained by using either symmetric rearrangement techniques (when $V$ is spherically symmetric) or Lions' concentration-compactness principle. In the problems considered in this article, the minimizing sequence $(u_k)$ is restricted to $H^1_\rho$, and in this case compactness is guaranteed through Strauss' radial lemma~\cite{Strauss}. More precisely, one has:

\begin{lemma}[Radial lemma]
\label{Lem:Radial}
\begin{enumerate}
\item[(a)] There is a constant $C > 0$ such that for all $u\in H^1_\rho$,
\begin{equation}
|u(x)| \leq \frac{C}{|x|}\| u \|_{H^1}
\quad\hbox{for almost all $x\in\Real^3$.}
\label{Eq:StraussIneq}
\end{equation}
\item[(b)] The space $H^1_\rho$ is compactly embedded in $L^p(\Real^3,\Complex^m)$ for all $p\in (2,6)$.
\end{enumerate}
\end{lemma}

\proof We use the fact that $u\in H^1(\Real^3,\Complex^m)$ implies $|u|\in H^1(\Real^3)$ and $\| |u| \|_{H^1}\leq \| u \|_{H^1}$ (see, for instance, Proposition~IX.3 in~\cite{Brezis}). Therefore, if $u\in H^1_\rho$, then $|u|\in H^1(\Real^3)$ is rotational symmetric and the inequality~(\ref{Eq:StraussIneq}) follows from the radial lemma with $m=1$~\cite{Strauss}.

To prove (b), suppose $(u_k)$ is a bounded sequence in $H^1_\rho$. Since $H^1_\rho$ is a Hilbert space, there exists a subsequence (denoted again by $(u_k)$) which converges weakly, $u_k\rightharpoonup u$, to some $u\in H^1_\rho$. Set $v_k := u_k - u\rightharpoonup 0$ and let $p\in (2,6)$. We claim that $(v_k)$ converges strongly to zero in $L^p(\Real^3,\Complex^m)$, which implies that $u_k\rightarrow u$ strongly in $L^p(\Real^3,\Complex^m)$. To prove the claim, let $\varepsilon > 0$ be arbitrary and set $K := \sup_{k\in \Natural} \| v_k \|^2_{H^1}$ which we can assume to be positive. According to~(\ref{Eq:StraussIneq}) there exists $R = R(\varepsilon) > 0$ such that
\begin{equation}
|v_k(x)|^{p-2}\leq \frac{\varepsilon}{K}
\quad\hbox{for almost all $x\in\Real^3\setminus B_R$ and all $k\in \Natural$},
\end{equation}
where $B_R$ denotes the ball of radius $R$ centered at the origin, such that
\begin{equation}
\int\limits_{\Real^3\setminus B_R} |v_k(x)|^p dx
\leq \frac{\varepsilon}{K}\int\limits_{\Real^3\setminus B_R} |v_k(x)|^2 dx\leq \varepsilon.
\end{equation}
Meanwhile, avoiding a contradiction with the Rellich-Kondrachov compactness theorem~\cite{Evans} requires that
\begin{equation}
\lim\limits_{k\to\infty}\int\limits_{B_R} |v_k(x)|^p dx = 0.
\end{equation}
Therefore,
\begin{equation}
\limsup_{k\to \infty} \int |v_k(x)|^p dx
 \leq \limsup_{k\to \infty}\int\limits_{B_R} |v_k(x)|^p dx + \limsup_{k\to \infty}\int\limits_{\Real^3\setminus B_R} |v_k(x)|^p dx
 \leq \varepsilon,
\end{equation}
and since $\varepsilon > 0$ was arbitrary, the claim follows.
\qed

\section{Existence of a global minimum in the irreducible case}
\label{Sec:Min}

With the preliminary results discussed in the previous section one can prove the following theorem which implies the first part of Theorem~\ref{Thm:Main}, regarding the existence of the minimizer. The second part of the Theorem, regarding the properties of the minimizers and their orbital stability are proven in Theorems~\ref{Thm:WeakSolution}, \ref{Thm:Strong}, and in Section~\ref{Sec:OrbitalStability}.

\begin{theorem}
\label{Thm:Existence}
Let $N > 0$ and assume the validity of conditions (i) and (ii) on $V$. Suppose $(u_k)$ is a sequence in $\mathcal{A}_{\rho,N} = \{ u \in H^1_\rho: \| u \|_2^2 = N \}$ which converges weakly to $u_*$ in $H^1_\rho$ and such that $\mathcal{E}[u_k]\to E_\ell(N)$. Then,
\begin{enumerate}
\item[(a)] $\| u_* \|_2^2 = N$, that is, $u_*\in \mathcal{A}_{\rho,N}$,
\item[(b)] $u_*$ is a minimizer: $\mathcal{E}[u_*] = E_{\ell}(N)$,
\item[(c)] $(u_k)$ converges strongly in $H^1_\rho$.
\end{enumerate}
\end{theorem}

\begin{remark}
In particular, the theorem implies that any minimizing sequence is relatively compact.
\end{remark}

\proofof{Theorem~\ref{Thm:Existence}} Using Lemma~\ref{Lem:Radial} and its proof, we have $u_k\rightharpoonup u_*$ weakly in $H^1_\rho$ and $u_k\rightarrow u_*$ strongly in $L^p(\Real^3,\Complex^m)$ for all $p\in (2,6)$. Consequently, $n_k := |u_k|^2\to  |u_*|^2 =: n_*$ strongly in $L^q(\Real^3,\Complex^m)$ for all $q\in (1,3)$. According to Lemma~\ref{Lem:BasicEstimates}(b) this implies
\begin{equation}
\lim\limits_{k\to\infty} D[n_k] = D[n_*].
\end{equation}
Next, we claim that also
\begin{equation}
\lim\limits_{k\to\infty} L[n_k] = L[n_*].
\end{equation}
For this, let $\varepsilon > 0$. Since $\| n_k \|_1 = N$ and $u_k\rightharpoonup u_*$ it follows that $\| n_* \|_1 = \| u_* \|_2^2\leq \liminf\limits_{k\to\infty} \| u_k \|_2^2 = N$. Hence, according to Lemma~\ref{Lem:BasicEstimates}(a),
\begin{equation}
\limsup_{k\to \infty}\left| L[n_k] - L[n_*] \right| \leq C_\varepsilon\limsup_{k\to \infty} \| n_k - n_* \|_2 + 2N\varepsilon = 2N\varepsilon,
\end{equation}
and since $\varepsilon > 0$ is arbitrary the claim follows.

Next, we use the fact that the Dirichlet functional is weakly lower semi-continuous (see Section~8.2 in Ref.~\cite{Evans}),
\begin{equation}
T[u_*]\leq \liminf_{k\to\infty} T[u_k],
\label{Eq:DirichletWLSC}
\end{equation}
to conclude that
\begin{equation}
\mathcal{E}[u_*] \leq \liminf_{k\to\infty} \mathcal{E}[u_k] = E_{\ell}(N).
\end{equation}

In a next step we show that $\| u_* \|_2^2 = N$, which implies statements (a) and (b) of the theorem since we then know that $E_{\ell}(N)\leq \mathcal{E}[u_*]$. To prove this claim, recall that $\tilde{N} := \| u_* \|_2^2\leq N$ and assume by contradiction that $\tilde{N} < N$. Clearly, $\tilde{N} > 0$ otherwise $u_* = 0$ which would contradict $\mathcal{E}[u_*] \leq E_{\ell}(N) < 0$. Set
\begin{equation}
v := \alpha u_* \in \mathcal{A}_{\rho,N},\qquad
\alpha := \sqrt{\frac{N}{\tilde{N}}} > 1.
\end{equation}
It follows that
\begin{equation}
\mathcal{E}[v] = \alpha^2\left( T[u_*] -  L[n_*] - \alpha^2 D[n_*] \right)
 < \alpha^2 \mathcal{E}[u_*] < E_{\ell}(N),
\end{equation}
which would contradict the definition of $E_\ell(N)$. Therefore, $\| u_* \|_2^2 = N$ as claimed.

To prove (c) we note that now that we know that $\mathcal{E}[u_*] = E_{\ell}(N)$, Eq.~(\ref{Eq:DirichletWLSC}) can be strengthened to
\begin{equation}
\lim_{k\to\infty} T[u_k] =  T[u_*],
\end{equation}
which, together with $\| u_k \|_2 = \| u_* \|_2 = \sqrt{N}$ implies that
\begin{equation}
\lim\limits_{k\to\infty} \| u_k \|_{H^1} = \| u_* \|_{H^1}.
\end{equation}
Since $u_k\rightharpoonup u_*$, it follows that $u_k\rightarrow u_*$ strongly in $H^1_\rho$ (see, for instance, Propositions~III.30 and V.1 in Ref.~\cite{Brezis}).
\qed

\begin{remark}
Note that, so far, we have not used the rotational symmetry of the potential $V$. However, this requirement will be needed in the next section to establish the fact that $u_*$ is a weak (and strong) solution of Eq.~(\ref{Eq:SPTimeIndependent}). The strong convergence of the minimizing sequence in $H^1_\rho$ will be important to show the orbital stability of the solution in section~\ref{Sec:OrbitalStability}.
\end{remark}

\section{Properties of the minimum in the irreducible case}
\label{Sec:Properties}

In this section we prove that a minimizer $u_*$ of $E_{\ell}(N)$ is a strong solution of Eq.~(\ref{Eq:SPTimeIndependent}) and discuss some of its properties. In a first step, we show that $u_*\in H^1_\rho$ is a weak solution of Eq.~(\ref{Eq:SPTimeIndependent}). For this, we first note the following lemma that will be useful:

\begin{lemma}
\label{Lem:GravPot}
Let $u\in H^1(\Real^3,\Complex^m)$. Then, the associated gravitational potential $U_u := \Delta^{-1}(|u|^2): \Real^3\to \Real$ is bounded, continuous, non-positive, and
\begin{equation}
\lim\limits_{|x|\to \infty} U_u(x) = 0.
\end{equation}
In particular, $U_u$ satisfies the same conditions (i) and (ii) as $V$.
\end{lemma}

\proof [cf. Lemma II.25 in Ref.~\cite{LiebSimon77}]
By Sobolev's inequality, $n:=|u|^2\in L^q(\Real^3)$ for all $1\leq q\leq 3$. Since $U_u = k * n$ is the convolution of
\begin{equation}
k(x) := -\frac{1}{4\pi|x|},\qquad x\in\Real^3\setminus \{ 0 \},
\end{equation}
with $n$, and since $k = k_1 + k_2$ with $k_1\in L^2(\Real^3)$ and $k_2\in L^4(\Real^3)$, it follows from H\"older's inequality that $U_u\in L^\infty(\Real)$ and that $\| U_u \|_\infty \leq \| k_1 \|_2 \| n \|_2 + \| k_2 \|_4 \| n \|_{4/3}$. Furthermore, since the space of compactly supported smooth functions $C_0^\infty(\Real^3)$ is dense in $L^q(\Real)$ for all $1\leq q < \infty$ it follows that $k_1$, $k_2$, and $n$  can be approximated by functions in $C_0^\infty(\Real^3)$. Therefore, $U_u$ can be approximated in $L^\infty(\Real^3)$ by functions in $C_0^\infty(\Real^3)$, which implies that it is continuous and converges to zero when $|x|\to \infty$. Finally, since $k\leq 0$ and $n\geq 0$ it follows that $U_u\leq 0$.
\qed

Next, we recall the notation $X = H^1(\Real^3,\Complex^m)$ and show:

\begin{lemma}
\label{Lem:Gateaux}
The energy functional $\mathcal{E}: X\to \Real$ defined in Eq.~(\ref{Eq:EnergyFunctional}) is G\^ateaux-differentiable:
\begin{equation}
\left. \frac{d}{dt}\mathcal{E}[u + t v] \right|_{t=0} = \re\left( B[u;v] \right),
\qquad u,v\in X,
\label{Eq:GateauxDeriv}
\end{equation}
where for each $u\in X$, $B[u; \cdot]: X\to \Complex$ is the bounded linear functional, given explicitly by
\begin{equation}
B[u;v] := \int \left[\overline{\nabla u(x)}\cdot \nabla v(x)+ (V(x) + U_u(x)) \overline{u(x)}\cdot v(x)  
\right] dx,\qquad v\in X.
\label{Eq:BDef}
\end{equation}
\end{lemma}

\proof Let $u,v\in X$ be fixed. It is simple to verify that $\mathcal{E}[u + tv]$ is a fourth-order polynomial in $t\in\Real$. Explicitly, one finds
\begin{equation}
\mathcal{E}[u + tv] - \mathcal{E}[u] = t\re\left( B[u;v] \right) + \frac{t^2}{2} S[u;v] + t^3\re\int \overline{u(x)} U_v(x) v(x) dx
- t^4 D[|v|^2],
\label{Eq:Eu+tv}
\end{equation}
with $B[u; v]$ as in Eq.~(\ref{Eq:BDef}) and
\begin{equation}
S[u;v] := \int \left[ | \nabla v(x)|^2 + (V(x) + U_u(x)) |v(x)|^2  \right] dx - 2D[2\re(\overline{u} v)],
\qquad v\in X.
\end{equation}
Due to Lemma~\ref{Lem:GravPot}, $V + U_u$ satisfies the same condition (i) as $V$, from which it follows easily that $B$ satisfies the desired conditions. Likewise, one verifies that $S$ and the remaining coefficients in Eq.~(\ref{Eq:Eu+tv}) are well-defined. Hence Eq.~(\ref{Eq:GateauxDeriv}) follows.
\qed

\begin{remark}
Suppose $u\in H^2(\Real^3,\Complex^m)$ satisfies the stationary problem~(\ref{Eq:SPTimeIndependent}) for some $\omega\in\Real$. Taking the $L^2$-scalar product of both sides of Eq.~(\ref{Eq:SPTimeIndependent}) with $v\in X$ yields
\begin{equation}
B[u; v] = \omega(u,v),\qquad v\in X,
\label{Eq:SPTTimeIndependentWeakForm}
\end{equation}
where $(u,v):=\int \overline{u(x)} v(x) dx$. We define a weak solution of Eq.~(\ref{Eq:SPTimeIndependent}) to be a function $u\in X$ such that Eq.~(\ref{Eq:SPTTimeIndependentWeakForm}) holds for all $v\in X$.
\end{remark}

The next result shows that a minimizer $u_*$ constructed in the previous section is a weak solution of Eq.~(\ref{Eq:SPTimeIndependent}) satisfying Eq.~(\ref{Eq:Virial2}).

\begin{theorem}
\label{Thm:WeakSolution}
Let $N > 0$ and suppose $V$ satisfies the conditions (i) and (ii) and is rotational symmetric. Then, a minimizer $u_*$ of $E_{\ell}(N)$ is a weak solution of Eq.~(\ref{Eq:SPTimeIndependent}. Moreover, 
\begin{equation}
\omega = \frac{2}{N}\left( E_{\ell}(N) - D[n_*] \right) < 0,\qquad n_* = |u_*|^2,
\label{Eq:omega}
\end{equation}
and
\begin{equation}
S[u_*,v] \geq \omega \| v \|_2^2
\label{Eq:SecondVariation}
\end{equation}
for all $v\in H^1_\rho$ satisfying $(v,u_*) = 0$.
\end{theorem}

\proof Let $u_*\in H^1_\rho$ be a minimizer and set $n_* := |u_*|^2$. In a first step we show that Eq.~(\ref{Eq:SPTTimeIndependentWeakForm}) holds for all $v\in H^1_\rho$. For this, let $v\in H^1_\rho$ and consider the variation
\begin{equation}
w(t) := \alpha(t) (u_* + t v),
\qquad \alpha(t) := \frac{\sqrt{N}}{\| u_* + t v \|_2},
\qquad |t| < \varepsilon,
\end{equation}
in $\mathcal{A}_{\rho,N}$, with $\varepsilon > 0$ small enough such that $u_* + t v\neq 0$ for all $t\in (-\varepsilon,\varepsilon)$. Using the minimizing property of $u_*$, the definition of $\mathcal{E}$, Leibnitz' rule, and Eq.~(\ref{Eq:GateauxDeriv}) one obtains, taking into account that $\alpha(0)=1$,
\begin{align}
0 = \left. \frac{d}{dt} \mathcal{E}[w(t)] \right|_{t=0}
 &= \left. \frac{d}{dt} \left\{ \alpha(t)^2\mathcal{E}[u_* + tv] + \alpha(t)^2(1-\alpha(t)^2) D[ |u_* + tw|^2 ]
 \right\} \right|_{t=0}
\nonumber\\
 &= \left. \frac{d}{dt} \alpha(t)^2 \right|_{t=0} \left( \mathcal{E}[u_*] - D[n_*] \right)
 + \left. \frac{d}{dt} \mathcal{E}[u_* + tv] \right|_{t=0}.
\end{align}
Since $\left. \frac{d}{dt} \alpha(t)^2 \right|_{t=0} = -2N^{-1}\re (u_*,v)$ and $\mathcal{E}[u_*] = E_\ell(N)$ one finds
\begin{equation}
 0 = \re\left(  B[u_*;v] - \omega (u_*,v) \right),
\label{Eq:StationaryPoint}
\end{equation}
with $\omega$ given by Eq.~(\ref{Eq:omega}). Interchanging $v$ with $i v$ in Eq.~(\ref{Eq:StationaryPoint}) yields Eq.~(\ref{Eq:SPTTimeIndependentWeakForm}) for $v\in H^1_\rho$.

Next, assume that $(v,u_*) = 0$, which implies that $\left. \frac{d}{dt} \alpha(t)^2 \right|_{t=0} = 0$ and $\left. \frac{d^2}{dt^2} \alpha(t)^2 \right|_{t=0} = -\frac{2}{N} \| v \|_2^2$. Using again the minimizing property of $u_*$ and Eqs.~(\ref{Eq:Eu+tv}) and (\ref{Eq:omega}) one finds
\begin{equation}
0\leq \left. \frac{d^2}{dt^2} \mathcal{E}[w(t)] \right|_{t=0}
 = S[u_*;v] - \omega \| v \|_2^2,
\end{equation}
which proves~(\ref{Eq:SecondVariation}).

In a final step we prove that Eq.~(\ref{Eq:SPTTimeIndependentWeakForm}) holds for arbitrary $v\in X$. To this purpose, we introduce the orthogonal projector $P: X\to X$ onto the space $H^1_\rho$ whose properties are analyzed in appendix~\ref{App:Projector}. In particular, it can be shown that $P$ leaves the scalar product of $L^2(\Real^3,\Complex^m)$ invariant and that $P$ commutes with the multiplication operator by any rotational-symmetric function. Since for any $u\in X$, $R\in SO(3)$, and $x\in\Real^3$,
\begin{equation}
U_u(R^{-1} x) = -\frac{1}{4\pi}\int \frac{|u(R^{-1}y)|^2}{|x-y|} dy
= -\frac{1}{4\pi}\int \frac{|[\rho(R) u](y)|^2}{|x-y|} dy,
\end{equation}
it follows that $U_u$ is rotional-symmetric for the minimizer $u=u_*\in H_\rho^1$. Therefore, under the assumption that $V$ is also rotational symmetric, the form $B: X\times X\to \Complex$ defined by Eq.~(\ref{Eq:BDef}) satisfies
\begin{equation}
B[u_*;v] = B[u_*;Pv],
\end{equation}
for all $v\in X$, since $Pu_* = u_*$. Hence, Eq.~(\ref{Eq:SPTTimeIndependentWeakForm}) holds for all $v\in X$ and the claim follows.
\qed

\begin{remark}
An alternative approach to establish that Eq.~(\ref{Eq:SPTTimeIndependentWeakForm}) holds for an arbitrary test function $v \in X$ relies on the principle of symmetric criticality~\cite{PalaisSymmetric}. By recognizing $\mathcal{A}_N = \{u \in H^1(\mathbb{R}^3, \mathbb{C}^{2\ell + 1}) : \|u\|_2^2 = N\}$ as a smooth $SO(3)$-manifold, the compactness of the Lie group $SO(3)$ allows one to invoke Theorem~5.4 of~\cite{PalaisSymmetric}. This result ensures that the minimizer $u_*$, obtained over the symmetric subspace $\mathcal{A}_{\rho, N}$, is in fact a critical point of the energy functional $\mathcal{E}$ over the entire manifold $\mathcal{A}_N$. Consequently, one could take directly $v\in X$ (instead of $v\in H^1_\rho$) in the proof of Theorem~\ref{Thm:WeakSolution} and conclude that Eq.~\eqref{Eq:StationaryPoint} is valid for all $v\in X$.
\end{remark}

Next, we prove that any weak solution is automatically a strong solution. For the following, it is convenient to abbreviate $L^2:=L^2(\Real^3,\Complex^m)$ and $H^s:=H^s(\Real^3,\Complex^m)$ for $s=1,2,3,\ldots$. We first use the following lemma:

\begin{lemma}
\label{Lem:H0Properties}
Suppose $V$ satisfies condition (i). Then, the linear operator $H_0 := -\Delta + V: D(H_0)\subset L^2\to L^2$ with domain $D(H_0) := H^2$ is self-adjoint and bounded from below.
\end{lemma}

\proof As a consequence of Lemma~\ref{Lem:Potential}, $V\in L^2(\Real^3) + L^\infty(\Real^3)$ which means that it lies in the Kato-Rellich class and implies that $H_0$ is self-adjoint in $L^2$ with domain $H^2$, see for instance chapter~V in~\cite{Kato-Book} or Theorem X.15 in~\cite{ReedSimon}. To prove the lower bound, we use Lemma~\ref{Lem:BasicEstimates}(c) to conclude that for each $u\in H^2$,
\begin{equation}
(u,H_0 u) = 2T[u] - 2L[n] \geq - C\| u \|_2^2,\qquad
n = |u|^2,
\end{equation}
which establishes the claim.
\qed

\begin{theorem}
\label{Thm:Strong}
Suppose the conditions (i) and (ii) on $V$ are satisfied and $u\in H^1(\Real^3,\Complex^m)$ is a weak solution of Eq.~(\ref{Eq:SPTimeIndependent}) with $\omega < 0$. Then, $u\in H^2(\Real^3,\Complex^m)$ and $u$ is strong solution of Eq.~(\ref{Eq:SPTimeIndependent}).

Moreover, if $V\equiv 0$, then $u\in C^\infty(\Real^3,\Complex^m)$ is a classical solution of Eq.~(\ref{Eq:SPTimeIndependent}).
\end{theorem}

\proof Since $V + U_u$ also satisfies the conditions (i) and (ii), the linear operator $H_u := -\Delta + V + U_u: D(H_0)\subset L^2\to L^2$ with domain $D(H_u) := H^2$ is self-adjoint and bounded from below, according to the previous lemma. Therefore, there exists a constant $\lambda > 0$ such that
\begin{equation}
\lambda + H_u : H^2\to L^2
\end{equation}
is invertible. By hypothesis, $u\in H^1$ is a weak solution of
\begin{equation}
(\lambda + H_u) u = G(u),\qquad
G(u) := (\lambda + \omega) u.
\end{equation}
Since $u\in H^1\subset L^2$, it follows that $u = (\lambda + H_u)^{-1} G(u)\in H^2$ satisfies Eq.~(\ref{Eq:SPTimeIndependent}) in the strong sense.

When $V\equiv 0$ we can use the fact that $1 - \Delta$ is an isomorphism from $H^{s+2}$ to $H^s$ and the fact that the map $F$ defined below Eq.~(\ref{Eq:Cauchy2}) satisfies $F: H^s\to H^s$ for all $s\geq 2$ (see~\cite{Lenzmann}). From these observations and $(1-\Delta) u = (1+\omega) u - F(u)$ it follows recursively that $u\in H^4$, $u\in H^6$, $u\in H^8$, ... and thus that $u\in C^\infty(\Real^3,\Complex^m)$ according to Sobolev's embedding theorems (see, for instance~\cite{Brezis}).
\qed

\begin{remark}
If $u$ is a minimizer, it follows from Eq.~(\ref{Eq:SecondVariation}) and the fact that $D[g]\geq 0$ for any function $g\in H^1$ that
\begin{equation}
(v,H_u v) = S[u; v] + 2D[2\re(\overline{u} v)] \geq \omega \| v \|_2^2
\end{equation}
for all $v\in H^2$ with $(v,u) = 0$. In particular, since $H_u u = \omega u$, it follows that $u$ is a ground state of the Schr\"odinger operator $H_u$.
\end{remark}

Finally, we prove Theorem~\ref{Thm:BHPotential} regarding the additional properties of the minimum when $V = V_{BH}$. For this, we make the dependency on the mass parameter $M$ in $V_{BH}$ explicit by writing $\mathcal{E}_M[u]$ and $L_M[n]$ instead of $\mathcal{E}[u]$ and $L[n]$. According to Theorem~\ref{Thm:Main}, there exists for each $M\geq 0$ and $N > 0$ a minimizer $u_*\in H_\rho^1$ such that $\|u_* \|_2^2 = N$ and $\mathcal{E}_M[u_*] = E_\ell(M,N)$.

To prove property (a), one uses the rescaling~(\ref{Eq:Rescaling}) and notes that when $V(x)$ is proportional to $1/|x|$, then $L_M[n_\lambda] = \lambda L_M[n]$. This implies
\begin{equation}
\mathcal{E}_M[u_\lambda] = \lambda^2 T[u] - \lambda L_M[n] - \lambda D[n],
\end{equation}
for all $\lambda > 0$. Taking $u = u_*$ to be a minimizer, differentiating both sides with respect to $\lambda$, and evaluating at $\lambda=1$ yields the desired result.

To prove the scaling properties in (b) notice that for any $\lambda > 0$,
\begin{equation}
\tilde{u}_\lambda(x) := \sqrt{\lambda} u_\lambda(x) = \lambda^2 u(\lambda x),
\end{equation}
satisfies $\| \tilde{u}_\lambda \|_2^2 = \lambda \| u \|_2^2$ and
\begin{equation}
\mathcal{E}_{\lambda M}[\tilde{u}_\lambda] = \lambda^3\mathcal{E}_M[u].
\end{equation}
Consequently, for all $\lambda > 0$,
\begin{align}
E_\ell(\lambda M,\lambda N) &= \inf\{ \mathcal{E}_{\lambda M}[u] : u\in H_\rho^1, \| u \|_2^2 = \lambda N \}\\
 &= \inf\{ \mathcal{E}_{\lambda M}[\tilde{u}_\lambda] : \tilde{u}_\lambda\in H_\rho^1, \| \tilde{u}_\lambda \|_2^2 = \lambda N \}\\
  &= \inf\{ \lambda^3\mathcal{E}_M[u] : u\in H_\rho^1, \| u \|_2^2 = N \} \\
  &= \lambda^3 E_\ell(M,N).
\end{align}
Setting $\lambda := N^{-1}$ yields the first equality in Eq.~(\ref{Eq:Scaling}). Furthermore, it follows that if $u_*$ is a minimizer for $N^{-3} E_\ell(N^{-1} M,1)$, then $\tilde{u}_{*N}$ is a minimizer for $E_\ell(M,N)$. To prove the second inequality it is sufficient to observe that
\begin{equation}
N\omega_\ell(M,N) = 2(E_\ell(M,N) - D[|\tilde{u}_{*N}|^2]) = 2N^3( E_\ell(N^{-1} M,1) - D[n_*]) = N^3\omega_\ell(N^{-1} M,1).
\end{equation}

For clarity of the presentation we break the proof of Theorem~\ref{Thm:BHPotential}(c) in a series of lemmas. First, we show:

\begin{lemma}
\label{Lem:Monotonicity}
The function $E_\ell: [0,\infty)\times (0,\infty)\to \Real$ is monotonously decreasing in both arguments.
\end{lemma}

\proof Let $M'\geq M\geq 0$. Then $\mathcal{E}_{M'}[u]\leq \mathcal{E}_M[u]$ for all $u\in H_\rho^1$, which implies $E_\ell(M',N)\leq E_\ell(M,N)$ and establishes the statement for the first argument of $E_\ell$. Next, let $\alpha\in (0,1]$. Then,
\begin{equation}
\mathcal{E}_M[\alpha u] = \alpha^2\mathcal{E}_M[u] + \alpha^2(1-\alpha^2) D[n]\geq \alpha^2\mathcal{E}_M[u]
\end{equation}
for all $u\in H_\rho^1$, which yields $E_\ell(M,\alpha^2 N)\geq \alpha^2 E_\ell(M,N)$. Therefore, if $0 < N' < N$, then
\begin{equation}
E_\ell(M,N') \geq \frac{N'}{N} E_\ell(M,N) > E_\ell(M,N), 
\end{equation}
since $N'/N < 1$ and $E_\ell(M,N) < 0$. It follows that $E_\ell(M,N)$ is strictly monotonously decreasing in $N$ (and that $N\mapsto N^{-1} E_\ell(M,N)$ is monotonously decreasing).
\qed

\begin{lemma}
Let $M > 0$. Then $\omega_\ell(M,N)\to E_\ell^{(0)}$ as $N\to 0$, where $E_\ell^{(0)}$ is the ground state energy of $H_0 = -\Delta + V_{BH}$.
\end{lemma}

\proof By definition,
\begin{equation}
E_\ell^{(0)} = \inf\left\{ \frac{(u,H_0 u)}{\| u \|_2^2} : u\in H_\rho^1, u\neq 0 \right\}.
\end{equation}
Since $2\mathcal{E}[u] = (u,H_0 u) - 2D[n]$ and $D[n]\geq 0$, it follows one one hand that
\begin{equation}
2E_\ell(M,N) \leq N E_\ell^{(0)},
\label{Eq:ErhoEstim1}
\end{equation}
and on the other hand that
\begin{equation}
2\mathcal{E}[u] - 2D[n] \geq N E_\ell^{(0)} - 4D[n],
\end{equation}
for all $u\in H^1_\rho$ with $\| u \|_2^2 = N$. By taking a minimizer, using Eq.~(\ref{Eq:DnInequality}), and part (a) of the theorem, which yields
\begin{equation}
D[n_*]\leq C N^{3/2} T[u_*]^{1/2} = C N^{3/2} \sqrt{|E_\ell(M,N)|},
\end{equation}
one finds
\begin{equation}
2E_\ell(M,N) - 2D[n_*]\geq N E_\ell^{(0)} - 4C N^{3/2}\sqrt{|E_\ell(M,N)|}.
\label{Eq:ErhoEstim2}
\end{equation}
We conclude from Eqs.~(\ref{Eq:omega}),(\ref{Eq:ErhoEstim1}), and (\ref{Eq:ErhoEstim2}) that
\begin{equation}
E_\ell^{(0)} - 4C\sqrt{N|E_\ell(M,N)|}
\leq \omega_\ell(M,N) \leq E_\ell^{(0)},
\end{equation}
for all $M\geq 0$ and $N > 0$. Since for fixed $M$, $|E_\ell(M,N)|$ is increasing in $N$ according to Lemma~\ref{Lem:Monotonicity}, it follows that $\sqrt{N|E_\ell(M,N)|}\to 0$ as $N\to 0$ and the claim follows.
\qed

In view of the scaling in part (b) it is sufficient to show the following statement to establish  the second limit in Eq.~(\ref{Eq:MLimits}) and conclude the proof of Theorem~\ref{Thm:BHPotential}:

\begin{lemma}
Let $M > 0$. Then $\omega_\ell(M,1)\to \omega_\ell(0,1)$ when $M\to 0$.
\end{lemma}

\proof Let $u_M\in H_\rho^1$, $\| u_M \|_2^2 = 1$, be a minimizer of $E_\ell(M,1)$. Then, since $M\geq 0$,
\begin{equation}
E_\ell(0,1)\leq \mathcal{E}_0[u_M] = T[u_M] - D[|u_M|^2]
 = E_\ell(M,1) + L_M[|u_M|^2].
\end{equation}
Therefore, taking into account the monotonicity of $E_\ell$ in the first argument, it follows that
\begin{equation}
-L_M[|u_M|^2] \leq E_\ell(M,1) - E_\ell(0,1) \leq 0,
\label{Eq:EellSandwich}
\end{equation}
for all $M\geq 0$. Using Hardy's inequality (see for example~\cite{Evans}) yields
\begin{equation}
L_M[|u_M|^2] = \frac{M}{2}\int \frac{|u_M(x)|^2}{|x|} dx \leq M\| u_M \|_2 \| \nabla u_M \|_2 = M\sqrt{2T[u_M]}.
\end{equation}
Since according to (a), $T[u_M] = |E_\ell(M,1)|$, one obtains for all $0\leq M\leq 1$ that
\begin{equation}
L_M[|u_M|^2] \leq M\sqrt{2|E_\ell(1,1)|}.
\label{Eq:LMBound}
\end{equation}
Equations~(\ref{Eq:EellSandwich}) and (\ref{Eq:LMBound}) imply that $E_\ell(M,1)\to E_\ell(0,1)$ for $M\to 0$. Since $\omega_\ell(M,1) = 2(3E_\ell(M,1) + L_M[|u_M|^2])$, the claim follows.
\qed

\section{Existence and properties of the global minimum in the reducible case}
\label{Sec:MutliEll}

In this section we prove the first part of Theorem~\ref{Thm:Main2}, regarding the existence and basic properties (a) and (b) of multi-$\ell$-boson stars. How to prove property (c) of the Theorem is explained in Section~\ref{Sec:OrbitalStability}.

First, it follows from Lemma~\ref{Lem:Coersive} and its proof that
\begin{equation}
-\infty < E(N) \leq E_\rho(N_1,N_2,\ldots,N_r) < 0,\qquad N = N_1 + N_2 + \ldots + N_r,
\end{equation}
since the rescaling~(\ref{Eq:Representation}) leaves $H_\rho^1$ and the subspace
\begin{equation}
\mathcal{A}_{\rho,N_1,N_2,\ldots,N_r} := \{ u\in H^1_\rho : \| u^{(j)} \|_2^2 = N_j \hbox{ for each } j=1,2,\ldots,r \}
\end{equation}
invariant. Let $(u_k)$ be a sequence in $\mathcal{A}_{\rho,N_1,N_2,\ldots,N_r}$ such that $\mathcal{E}[u_k]\to E_\rho(N_1,N_2,\ldots,N_r)$. It follows from Lemma~\ref{Lem:Coersive}(b) that $(u_k)$ is bounded in $H^1_\rho$, and hence by Lemma~\ref{Lem:Radial} we have (after selecting a  subsequence) $u_k\rightharpoonup u_*$ weakly in $H^1_\rho$ and $u_k\to u_*$ strongly in $L^p(\Real^3,\Complex^m)$ for all $p\in (2,6)$. As in the proof of Theorem~\ref{Thm:Existence} it follows that $\mathcal{E}[u_*]\leq E_\rho(N_1,N_2,\ldots,N_r)$. Let $\tilde{N}_j := \| u_*^{(j)} \|_2^2$ for $j=1,2,\ldots,r$. By the properties of weak convergence, $\tilde{N}_j\leq N_j$, and to conclude that $u_*\in H^1_\rho$ is a minimizer it remains to show that equality holds for all $j=1,2,\ldots,r$. To prove this, we use the following

\begin{lemma}
\label{Lem:EnergyDifference}
Let $u = (u^{(1)},u^{(2)},\ldots,u^{(r)})\in H^1_\rho$ and $u'\in H^1_\rho$ be equal to $u$ except that one of the components $u^{(j)}$ is replaced with $w$. Then,
\begin{equation}
\mathcal{E}[u'] - \mathcal{E}[u] = \mathcal{E}_j[w] - \mathcal{E}_j[u^{(j)}],
\label{Eq:EnergyDifference}
\end{equation}
where
\begin{equation}
\mathcal{E}_j[w] := \frac{1}{2}\int \left( |\nabla w(x)|^2 + W_j(x)|w(x)|^2 \right) dx
 - \frac{1}{16\pi}\int\int \frac{|w(x)|^2 |w(y)|^2}{|x-y|} dx dy,
\label{Eq:Ej}
\end{equation}
and
\begin{equation}
W_j(x) := V(x) - \sum\limits_{j'\neq j}\frac{1}{4\pi}\int \frac{|u^{(j')}(y)|^2}{|x-y|} dy.
\end{equation}
\end{lemma}

\proof Direct verification.
\qed

\begin{remark}
$W_j$ is the sum of the external potential $V$ and the potential associated with the gravitational force acting on the $j$'th component due to the presence of the remaining components.
\end{remark}

\begin{corollary}
Suppose $u,u'\in H^1_\rho$ are as in Lemma~\ref{Lem:EnergyDifference} and set $\tilde{N}_j := \| u^{(j)} \|_2^2$. Then, for each $N_j > \tilde{N}_j$ one can choose $w$ in $u'$ such that $\| w \|_2^2 = N_j$ and $\mathcal{E}[u'] < \mathcal{E}[u]$.
\end{corollary}

\proof According to Lemma~\ref{Lem:GravPot} $W_j$ satisfies the conditions (i) and (ii), thus $\mathcal{E}_j: H^1_{\ell_j}\to \Real$ possesses the same qualitative properties as the energy functional $\mathcal{E}$. In particular, $\mathcal{E}_j$ has a minimum $w_*$ on $\{ w\in H^1_{\ell_j} : \| w \|_2^2 = N_j \}$ which satisfies $\mathcal{E}_j[w_*] < 0$. If $\mathcal{E}_j[u^{(j)}]\geq 0$ the statement is obvious. Hence, suppose $\mathcal{E}_j[u^{(j)}] < 0$. Then, setting $\alpha:=\sqrt{N_j/\tilde{N}_j} > 1$ such that $\| \alpha u^{(j)} \|_2^2 = N_j$ one obtains
\begin{equation}
\mathcal{E}_j[w_*] \leq \mathcal{E}_j[\alpha u^{(j)}] \leq \alpha^2 \mathcal{E}_j[u^{(j)}] < \mathcal{E}_j[u^{(j)}],
\end{equation}
and the corollary follows from Eq.~(\ref{Eq:EnergyDifference}).
\qed

Coming back to the proof of Theorem~\ref{Thm:Main2}, suppose $\tilde{N}_j < N_j$ for some $j\in \{1,2,\ldots,r \}$. Then, applying the above corollary we can substitute $u_*^{(j)}$ for some $w\in H^1_{\ell_j}$ with $\| w \|_2^2 = N_j$ such that the resulting vector $u\in H^1_\rho$ satisfies $\mathcal{E}[u] < \mathcal{E}[u_*]$. Repeating the process until $\| u^{(j)} \|_2^2 = N_j$ for all $j=1,2,\ldots,r$ yields $u\in \mathcal{A}_{\rho,N_1,N_2,\ldots,N_r}$ satisfying $\mathcal{E}[u] < \mathcal{E}[u_*]\leq E_\rho(N_1,N_2,\ldots,N_r)$, which contradicts the definition of $E_\rho(N_1,N_2,\ldots,N_r)$. Therefore, we conclude that $u_*\in \mathcal{A}_{\rho,N_1,N_2,\ldots,N_r}$ and $u_*$ is a minimizer. Furthermore, as in the proof of Theorem~\ref{Thm:Existence}, it follows that $u_k\to u_*$ strongly in $H^1_\rho$.

To prove properties (a) and (b) of Theorem~\ref{Thm:Main2}, we first show that $u_*$ is a weak solution of Eq.~(\ref{Eq:SPTimeIndependent}). For this, let $v = (v^{(1)},v^{(2)},\ldots,v^{(r)})\in H^1_\rho$ and consider the variation $w(t) = (w^{(1)}(t),w^{(2)}(t),\ldots,w^{(r)}(t))\in \mathcal{A}_{\rho,N_1,N_2,\ldots,N_r}$ defined by
\begin{equation}
w^{(j)}(t) := \alpha_j(t)\left (u_*^{(j)} + t v^{(j)} \right),\qquad
\alpha_j(t):=\frac{\sqrt{N_j}}{\| u_*^{(j)} + t v^{(j)} \|_2},\qquad
j=1,2,\ldots,r,
\end{equation}
for small enough $|t|$. Then, using Lemma~\ref{Lem:Gateaux} a similar calculation than the one used in the proof of Theorem~\ref{Thm:WeakSolution} yields
\begin{equation}
0 = \left. \frac{d}{dt} \mathcal{E}[w(t)] \right|_{t=0} = \re\left( B[u_*;v] - \sum\limits_{j=1}^r \omega_j(u_*,v^{(j)}) \right),
\end{equation}
with
\begin{equation}
\omega_j = \frac{2}{N_j}\left( \mathcal{E}_j[u_*^{(j)}] - \frac{1}{16\pi}\int\int \frac{|u_*^{(j)}(x)|^2 |u_*^{(j)}(y)|^2}{|x-y|} dx dy \right),
\label{Eq:omegaj}
\end{equation}
where $\mathcal{E}_j$ is as in Eq.~(\ref{Eq:Ej}). Replacing $v$ with $iv$ we conclude that for all $v\in H^1_\rho$
\begin{equation}
B[u_*;v] = \sum\limits_{j=1}^r \omega_j(u_*,v^{(j)}).
\end{equation}
Using the projector $P$ from appendix~\ref{App:Projector} one concludes from this that $u_*\in H^1_\rho$ is a weak solution of Eq.~(\ref{Eq:SPTimeIndependent}) with $\omega$ replaced with the matrix in Eq.~(\ref{Eq:OmegaMatrix}). As in the proof of Theorem~\ref{Thm:Strong} one concludes that $u_*\in H^2$ is also a strong solution. Finally, by multiplying Eq.~(\ref{Eq:omegaj}) with $N_j$ and summing over $j$, the relation~(\ref{Eq:Virial2Bis}) follows.

\section{Orbital stability}
\label{Sec:OrbitalStability}

In this section we prove the orbital stability of the sets $S_{\ell,N}$ and $S_{\rho,N_1,N_2,\ldots,N_r}$ defined in Eqs.~(\ref{Eq:SrhoN}) and (\ref{Eq:SrhoNs}), respectively, for perturbations $\psi_0\in H^2_\rho := H^1_\rho\cap H^2(\Real^3,\Complex^m)$. The argument is standard, see for instance~\cite{Cazenave1982,Frohlich:2005sh}, so we only provide the main steps and restrict ourselves to $S_{\ell,N}$  since the other case is completely analogous.

First, notice that for rotational-symmetric $V$, $\psi_0\in H^2_\rho$ implies that $\psi(t)\in H^2_\rho$ for all $t\in\Real$. This can be seen by replacing the spaces $L^2$ and $H^2$ with their spherically symmetric subspaces $L^2_\rho$ and $H^2_\rho$ in Theorem~\ref{Thm:WP} and exploiting the fact that both $H_0$ and $F$ commute with $\rho(R)$. The claim than follows by uniqueness of the solution.

Next, we note that according to Eq.~(\ref{Eq:CoersiveBound}) the set $S_{\ell,N}$ is uniformly bounded in $H^1_\rho$ by a constant that depends only on $N$. Suppose by contradiction that $S_{\ell,N}$ is not orbital stable. Then, there exists $\varepsilon > 0$, a sequence $(\psi_0^{(j)})$ in $H^1_\rho\cap H^2(\Real^3,\Complex^m)$ of initial data, and a sequence $(t_j)$ in $(0,\infty)$ such that
\begin{equation}
\lim\limits_{j\to\infty}\mbox{dist}(\psi_0^{(j)},S_{\ell,N}) = 0
\label{Eq:distpsi0}
\end{equation}
and the corresponding sequence of solutions $(\psi^{(j)}(t))$ of the Cauchy problem~(\ref{Eq:Cauchy1},\ref{Eq:Cauchy2}) satisfies
\begin{equation}
\mbox{dist}(\psi^{(j)}(t_j),S_{\ell,N}) \geq \varepsilon
\label{Eq:distpsij}
\end{equation}
for all $j=1,2,3,\ldots$

According to Eq.~(\ref{Eq:distpsi0}) there is a sequence $(v_j)$ in $S_{\ell,N}$ such that $\| \psi_0^{(j)} - v_j \|_{H^1}\to 0$ for $j\to \infty$. Therefore, using the fact that $\psi^{(j)}(t)$ preserves the $L^2$-norm,
\begin{equation}
\left| \| \psi^{(j)}(t_j) \|_2 - \sqrt{N} \right|
 = \left| \| \psi_0^{(j)} \|_2 - \| v_j \|_2 \right| \leq \| \psi_0^{(j)} - v_j \|_2 \to 0
\end{equation}
for $j\to\infty$. Likewise, by the local Lipschitz continuity of $\mathcal{E}$ (see Eq.~(\ref{Eq:LocalLipschitz})) and energy conservation, one obtains
\begin{equation}
\left| \mathcal{E}[ \psi^{(j)}(t_j)] - E_\ell(N) \right| = \left| \mathcal{E}[ \psi_0^{(j)} ] - \mathcal{E}[v_j] \right| \to 0
\end{equation}
for $j\to \infty$. Therefore, the sequence $(\tilde{u}_j):= (\psi^{(j)}(t_j))$ in $H^1_\rho$ satisfies $\| \tilde{u}_j \|_2^2 \to N$ and $\mathcal{E}[\tilde{u}_j] \to E_\ell(N)$.

If $\|\tilde{u}_j \|_2^2 = N$ for all $j\in \Natural$ this would imply that $(\tilde{u}_j)$ is a minimizing sequence in $\mathcal{A}_{\rho,N}$, and according to Theorem~\ref{Thm:GroundState}, it would have a subsequence which converges strongly in $H^1_\rho$ to a minimizer of $E_\ell(N)$ contradicting Eq.~(\ref{Eq:distpsij}). Otherwise we normalize $u_j := \alpha_j\tilde{u}_j$ with $\alpha_j := \sqrt{N}/\| \tilde{u}_j \|_2$, such that $u_j\in \mathcal{A}_{\rho,N}$. Since $\alpha_j\to 1$ one finds again $\mathcal{E}[u_j]\to E_\ell(N)$ which shows that $(u_j)$ is a minimizing sequence and leads to a contradiction with Eq.~(\ref{Eq:distpsij}).


\begin{acknowledgments}
We thank Piotr Chru\'sciel for suggesting the variational approach and Alberto Diez-Tejedor and Armando Roque for fruitful discussions. We also thank Alberto Diez-Tejedor and H\r{a}kan Andr{\'{e}}asson for comments on a previous version of the manuscript. E.C.N. was supported by a SECIHTI predoctoral scholarship. O.S. acknowledges support from SECIHTI-SNII and CIC grant No.~18315 to Universidad Michoacana de San Nicolás de Hidalgo.
\end{acknowledgments}

\appendix
\section{Global existence of the Cauchy problem}
\label{App:GlobalExistence}

In this appendix we provide a proof of Theorem~\ref{Thm:WP} based on semi-group methods (see for example~\cite{Pazy-Book,ReedSimon}). As before, we abbreviate $L^2 := L^2(\Real^3,\Complex^m)$ and $H^s := H^s(\Real^3,\Complex^m)$ for $s=1,2,3,\ldots$. We first recall that the Schr\"odinger operator $H_0 = -\Delta + V: D(H_0)\in L^2\to L^2$ is self-adjoint with domain $D(H_0) = H^2$, see Lemma~\ref{Lem:H0Properties}. According to Stone's theorem, the anti-selfadjoint operator $-i H_0: D(H_0)\subset L^2\to L^2$ generates a strongly continuous one-parameter unitary group $U(t) = \exp(-iH_0 t)$ in $L^2$. Next, we note that by Duhamel's formula, a solution of Eqs.~(\ref{Eq:Cauchy1},\ref{Eq:Cauchy2}) must satisfy
\begin{equation}
\psi(t) = U(t)\psi_0 + \int\limits_0^t U(t-s) F(\psi(s)) ds,\qquad
t\in\Real.
\label{Eq:IntegralRepresentation}
\end{equation}
Local in time well-posedness\footnote{See, for instance, Theorem~1.7 and the comments following it in chapter~6 of reference~\cite{Pazy-Book}. Note that in this reference strong solutions are called ``classical solutions''.} follows if one can show that the map $F: H^2\to H^2$ is locally Lipschitz continuous. This means that for each $\psi_0\in D(H_0)$ there exists a local strong solution, that is, a continuously differentiable curve $\psi: (t_1,t_2)\to L^2$ defined on a finite time interval $(t_1,t_2)$ containing $0$, such that $\psi(0) = \psi_0$, $\psi: (t_1,t_2)\to D(H_0)$ is continuous, and $\psi(t)$ satisfies Eq.~(\ref{Eq:Cauchy1}) for all $t\in (t_1,t_2)$. Local Lipschitz continuity of $F$ is implied by the following Lemma, which is a direct generalization to $m > 1$ of Lemma 3 in Ref.~\cite{Lenzmann} when specialized to $\mu=0$ and $s\in\Natural$:

\begin{lemma}
\label{Lem:KeyEstimate}
For each $s\in\Natural$ the map $F: H^s\to H^s$, $u\mapsto \Delta^{-1}(|u|^2) u$, is well-defined and satisfies the following inequalities: there are positive constants $C_1$ and $C_2$ depending only on $s$ such that for all $u,v\in H^s$,
\begin{enumerate}
\item[(a)] $\| F(u) \|_{H^s}\leq C_1 \| u \|_{H^r}^2 \| u \|_{H^s}$, where $r = \max\{ 1,s-1 \}$,
\item[(b)] $\| F(u) - F(v) \|_{H^s} \leq C_2\left( \| u \|_{H^s}^2 + \| v \|_{H^s}^2 \right)\| u - v \|_{H^s}$.
\end{enumerate}
\end{lemma}

\proof One can easily verify that the proof in~\cite{Lenzmann} goes through when $m>1$.
\qed

In order to prove global well-posedness we first show that a local strong solution conserves the total energy and the quantity $Q$ defined in Eq.~(\ref{Eq:QDef}).

\begin{lemma}
\label{Lem:EQConservation}
Let $t_1 < 0 < t_2$, and let $\psi: (t_1,t_2)\to L^2$ be a local strong solution of the Cauchy problem~(\ref{Eq:Cauchy1},\ref{Eq:Cauchy2}). Then, $\mathcal{E}[\psi(t)] = \mathcal{E}[\psi_0]$ and $Q[\psi(t)] = Q[\psi_0]$ for all $t\in (t_1,t_2)$.
\end{lemma}

\proof
We prove that the derivatives of $\mathcal{E}[\psi(t)]$ and $Q[\psi(t)]$ with respect to $t$ are equal to zero. First, note that the energy functional can be written as follows:
\begin{equation}
\mathcal{E}[\psi(t)] = \frac{1}{2} \left(\psi(t), H_0\psi(t)\right) + \frac{1}{4}\left(\psi(t), \Delta^{-1}(|\psi(t)|^2)\psi(t)\right).
\end{equation}
Its derivative is given by 
\begin{align}
\frac{d}{dt} \mathcal{E}[\psi(t)] &= \frac{1}{2}\lim_{s \to 0} \frac{1}{s}\left[\left(\psi(t + s), H_0 \psi(t + s)\right) - \left((\psi(t), H_0\psi(t)\right)\right]
\nonumber\\
&+ \frac{1}{4} \lim_{s \to 0} \frac{1}{s} \left[\left(\psi(t + s), \Delta^{-1}(|\psi(t + s)|^2)\psi(t + s) \right) - \left(\psi(t), \Delta^{-1}(|\psi(t)|^2)\psi(t)\right)\right].
\end{align}
After some algebra one obtains
\begin{align}
\frac{d}{dt} \mathcal{E}[\psi(t)] &= \frac{1}{2}\lim_{s \to 0} \left(\frac{\psi(t + s) - \psi(t)}{s}, H_0\psi(t)\right) + \frac{1}{2}\lim_{s \to 0} \left( H_0\psi(t + s), \frac{\psi(t + s) - \psi(t)}{s}\right)
\nonumber\\
&+ \frac{1}{4}\lim_{s \to 0} \left(\frac{\psi(t + s) - \psi(t)}{s}, [\Delta^{-1}(|\psi(t + s)|^2 + |\psi(t)|^2)]\psi(t)\right)
\nonumber\\
&+ \frac{1}{4}\lim_{s \to 0} \left([\Delta^{-1}(|\psi(t + s)|^2 + |\psi(t)|^2)]\psi(t + s), \frac{\psi(t + s) - \psi(t)}{s}\right),
\label{Eq:TimeDerivativeEnergy}
\end{align}
where we have used the facts that $H_0$ is self-adjoint and that $(\psi,\Delta^{-1}(|\phi|^2)\psi$ is symmetric in $\psi\leftrightarrow\phi$. By assumption,
$$
\lim_{s \to 0}\frac{\psi(t + s) - \psi(t)}{s} = \frac{d\psi}{dt}(t),\qquad
\lim_{s \to 0} H_0\psi(t+s) = H_0\psi(t),
$$
in $L^2$. Furthermore, using similar arguments than in the proof of Lemma~\ref{Lem:GravPot}, one finds
\begin{equation}
\|\Delta^{-1}(|\psi(t + s)|^2 -|\psi(t)|^2)\|_\infty 
\le C(\|\psi(t + s) + \psi(t)\|_4 \|\psi(t + s) - \psi(t)\|_4 
 + \|\psi(t + s) + \psi(t)\|_{8/3} \|\psi(t+ s) - \psi(t)\|_{8/3}),
\end{equation}
which implies $\Delta^{-1}(|\psi(t + s)|^2 + |\psi(t)|^2) \to 2\Delta^{-1}(|\psi(t)|^2$ in $L^\infty$, as $s\to 0$. Gathering the results, one obtains
\begin{equation}
\frac{d}{dt} \mathcal{E}[\psi(t)] = \text{Re} \left(\frac{d\psi(t)}{dt}, H_0\psi(t) + F(\psi(t))\right).
\end{equation}
Thus, from Eq.~\eqref{Eq:Cauchy1} it follows
\begin{equation}
 \frac{d}{dt} \mathcal{E}[\psi(t)] = \text{Re}\left(i\left\|\frac{d\psi(t)}{dt}\right\|_{2}^2\right) = 0. 
\end{equation}

To establish the analogous result for the   functional $Q[\psi(t)]$ it is convenient to consider its components, given by
\begin{equation}
Q_{jk}[\psi(t)] = \left(\psi_j(t), \psi_k(t)\right),\qquad
1\le j, k \le m.
\end{equation}
Using Eq.~\eqref{Eq:Cauchy1} and the self-adjointness of the operator $H_0 + \Delta^{-1}(|\psi(t)|^2)$ one finds
\begin{align}
\frac{d}{dt} Q_{jk}[\psi(t)] &= \left(\psi_j(t), \frac{d\psi_k(t)}{dt}\right) + \left(\frac{d\psi_j(t)}{dt}, \psi_k(t)\right)\\
 &= i\left(H_0\psi_j(t) + \Delta^{-1}(|\psi(t)|^2)\psi_j(t), \psi_k(t)\right) - i\left( H_0\psi_j(t) + + \Delta^{-1}(|\psi(t)|^2)\psi_j(t), \psi_k(t) \right) \\
 & = 0,
\end{align}
which concludes the proof of the lemma.
\qed

An important consequence of Lemma~\ref{Lem:EQConservation} is that the $H^1$-norm of a local solution is uniformly bounded. This is due to the bounds~(c) and (d) in Lemma~\ref{Lem:BasicEstimates} which imply that (see Eq.~(\ref{Eq:CoersiveBound}))
\begin{equation}
\| u \|_{H^1}^2\leq C\left( \mathcal{E}[u] + N + N^3 \right),
\qquad 
N = \Tr(Q[u]),\qquad
u\in H^1,
\end{equation}
such that $\| \psi(t) \|_{H^1}$ is bounded by a constant depending only on $\mathcal{E}[\psi_0]$ and $\Tr( Q[\psi_0])$. As a consequence of Eq.~(\ref{Eq:IntegralRepresentation}) and the fact that the $H^2$-norm is equivalent to the graph norm associated with $H_0$ and that $U(t)$ commutes with $H_0$, it follows that
\begin{equation}
\| \psi(t) \|_{H^2} \leq C\left( \| \psi_0 \|_{H^2} + \int\limits_0^t \| F(\psi(s)) \|_{H^2} ds \right)
 \leq C\left( \| \psi_0 \|_{H^2} + \int\limits_0^t \| \psi(s) \|_{H^2} ds \right),
\end{equation}
for all $0\leq t < t_2$, where we have used Lemma~\ref{Lem:KeyEstimate}(a) in the last step. Using Gronwall's inequality yields a uniform bound on the $H^2$-norm of $\psi(t)$. Therefore, the solution can be extended to $t_2=\infty$. Likewise, one can extend the solution to the past to $t_1=-\infty$. Finally, observe that by Sobolev's inequality, the uniform bound on $H^2$ also implies that $\| \psi(t) \|_\infty$ is uniformly bounded.

\section{Projector on the space of spherically symmetric functions}
\label{App:Projector}

We use the same abbreviations $L^2$ and $H^s$ as in the previous appendix. Let $\rho: SO(3)\to GL(L^2)$ be a unitary representation of the rotation group of the form given in Eq.~(\ref{Eq:Representation}), and denote by $P: L^2\to L^2$ the orthogonal projector onto the closed subspace
\begin{equation}
L_\rho^2 := \{ u\in L^2 : \rho(R)u = u \hbox{ for all $R\in SO(3)$} \}.
\end{equation}
In this appendix, we prove:

\begin{proposition}
\label{Prop:Projector}
The projector $P: L^2\to L^2$ defined above satisfies the following properties:
\begin{enumerate}
\item[(a)] $P(H^1) = H_\rho^1$,
\item[(b)] $(\nabla Pu,\nabla v) = (\nabla u,\nabla Pv)$ for all $u,v\in H^1$,
\item[(c)] Let $f: \Real^3\setminus \{0 \} \to \Real$ be a measurable rotational-symmetric function, and let $u\in L^2$ be such that $f u\in L^2$. Then, $f P u\in L^2$ and $P(fu) = f Pu$.
\end{enumerate}
\end{proposition}

\begin{remark}
Properties (a) and (b) imply that $(Pu,v)_{H^1} = (Pu,v) + (\nabla P u,\nabla v) = (u,Pv)_{H^1}$ for all $u,v\in H^1$ and that the image of the restriction $P_1: H^1\to H^1$, $u\mapsto P u$, is equal to $H_\rho^1$. Therefore, $P_1$ coincides with the orthogonal projector of $H^1$ onto $H_\rho^1$.
\end{remark}

To prove the proposition, we use the Fourier transformation $\mathcal{F}: L^2\to L^2$, defined as
\begin{equation}
[\mathcal{F} u](k) := \frac{1}{(2\pi)^{3/2}} \int u(x) e^{-ikx} dx,\qquad k\in \Real^3,
\end{equation}
for $u$ in a suitable dense subspace of $L^2$ (for example $C_0^\infty(\Real^3,\Complex^m)$), and first prove the following properties:

\begin{lemma}
\label{Lem:Fourier}
\begin{enumerate}
\item[(a)] $[\rho(R),\mathcal{F}] = 0$ for all $R\in SO(3)$,
\item[(b)] $\mathcal{F}$ and $\mathcal{F}^{-1}$ leave $L_\rho^2$ and its orthogonal complement invariant,
\item[(c)] $[P,\mathcal{F}] = 0$.
\end{enumerate}
\end{lemma}

\proof 
Using Eq.~(\ref{Eq:Representation}) one finds
\begin{align}
[\rho(R)\mathcal{F}u](k) 
 &= \frac{1}{(2\pi)^{3/2}}\int D(R) u(x) e^{-i(R^{-1} k)x} dx 
\nonumber\\
 &= \frac{1}{(2\pi)^{3/2}}\int D(R) u(R^{-1} y) e^{-ik y} dy
 = [\mathcal{F}\rho(R)u](k)
\end{align}
for all $k\in \Real^3$ and all $u\in C_0^\infty(\Real^3,\Complex^m)$, where we have used the variable substitution $x = R^{-1} y$ in the second step. Since $C_0^\infty(\Real^3,\Complex^m)$ is dense in $L^2$ and $\rho(R)$ and $\mathcal{F}$ are bounded, this proves (a).

Next, to prove (b), let $u\in L_\rho^2$. Then $\rho(R)\mathcal{F} u = \mathcal{F}\rho(R) u = \mathcal{F} u$ for all $R\in SO(3)$, which implies that $\mathcal{F} u\in L_\rho^2$. Likewise, it follows that $\mathcal{F}^{-1} u\in L_\rho^2$. Next, let $v\in (L_\rho^2)^\perp$. Using Parseval's identity, one finds
\begin{equation}
(w,\mathcal{F} v) = (\mathcal{F}^{-1} w,v) = 0
\end{equation}
for all $w\in L_\rho^2$, which implies that $\mathcal{F}v\in (L_\rho^2)^\perp$. Likewise, one concludes that $\mathcal{F}^{-1}v\in (L_\rho^2)^\perp$.

Finally, to prove (c), let $u\in L^2$. Then, using the orthogonal decomposition $u = Pu + u^\perp$ with $u^\perp\in (L_\rho^2)^\perp$, one finds
\begin{equation}
P\mathcal{F} u = P\mathcal{F}(Pu + u^\perp) = P\mathcal{F} P u + P(\mathcal{F} u^\perp) = \mathcal{F} P u,
\end{equation}
since $\mathcal{F} Pu\in L_\rho^2$ and $\mathcal{F} u^\perp\in (L_\rho^2)^\perp$ according to (b).
\qed

\proofof{Proposition~\ref{Prop:Projector}} We start with the proof of statement (c). For this, note first that using Eq.~(\ref{Eq:Representation}),
\begin{equation}
[\rho(R) f u](x) = f(R^{-1} x) D(R) u(R^{-1} x) = f(x) [\rho(R) u](x) = f(x) u(x),
\end{equation}
for all $x\in \Real^3$, $R\in SO(3)$ and $u\in L_\rho^2$. Therefore, $f u\in L_\rho^2$ if $f u\in L^2$ and $u\in L_\rho^2$. Likewise, one shows that $f u\in (L_\rho^2)^\perp$ if $f u\in L^2$ and $u\in (L_\rho^2)^\perp$. Therefore, as in the proof of Lemma~\ref{Lem:Fourier}(c), it follows that
\begin{equation}
P(f u) = P f(Pu + u^\perp) = f P u,
\end{equation}
for all $u\in L^2$ such that $f u\in L^2$.

To prove (a) recall that for any $s\geq 0$,
\begin{equation}
H^s = \{ u\in L^2 : (1+k^2)^{s/2}\mathcal{F}{u} \in L^2 \}.
\end{equation}
Let $u\in H^s$. Then, it follows using Lemma~\ref{Lem:Fourier}(c) and the already established property (c) of the proposition that
\begin{equation}
(1+k^2)^{s/2}\mathcal{F} P u = P(1+k^2)^{s/2} \mathcal{F} u\in L^2;
\end{equation}
hence $Pu\in H^s$. It follows that $P$ maps $H^s$ onto $H_\rho^s := H^s\cap L_\rho^2$ for each $s\geq 0$, in particular, $P(H^1) = H_\rho^1$.

Finally, to prove (b), let $u,v\in H^1$. Using Parseval's identity and the previous results, one finds
\begin{align*}
(\nabla P u,\nabla v) &= (\mathcal{F}\nabla P u,\mathcal{F}\nabla v) \\
 &= (|k|\mathcal{F} P u,|k|\mathcal{F} v) \\
 &= (P|k|\mathcal{F} u,|k|\mathcal{F} v) \\
 &= (|k|\mathcal{F} u,|k| P\mathcal{F} v ) \\
 &= (\nabla u,\nabla P v),
\end{align*}
which concludes the proof of the Proposition.
\qed

\bibliography{references.bib} 

\end{document}